\documentclass{article}
\usepackage{setspace}
\usepackage{braket}
\usepackage{bbm}
\usepackage{tcolorbox}
\usepackage{graphicx}
\usepackage{amsthm}
\usepackage{hyperref}
\hypersetup{
    colorlinks,
    citecolor= orange,
    filecolor= blue,
    linkcolor= blue,
    urlcolor= green
}
\usepackage{multicol}
\usepackage{mathrsfs}
\usepackage{appendix}
\usepackage{amssymb}
\usepackage{relsize}
\usepackage{amsthm}
\usepackage{ amssymb }
\usepackage{amsmath}
\usepackage[utf8]{inputenc}
\usepackage[english]{babel}
\usepackage[a4paper, total={6in, 9.2in}]{geometry}
\usepackage{amsmath}

\newtheorem{definition}{Definition}
\newtheorem{definition2}{Lemma}

\newtheorem{definition4}{Theorem}

\newtheorem{Co}{Corollary}
\usepackage[T1]{fontenc}
\usepackage{lmodern}

\usepackage{tikz-cd}

\begin{titlepage}

\title{SBSCV}
\author{Alberto Acevedo}

\date{2023}

\end{titlepage}

\begin{document}

\thispagestyle{plain}
\begin{center}
   
    \textbf{Monitoring the process of system-information broadcasting in time for continuous variables}

    \vspace{0.4cm}
    
    \textbf{Alberto Acevedo$^{1,3}$, Jarek Korbicz$^{2}$, Janek Wehr $^{1}$ }

      \vspace{0.4cm}
      \textbf{$^{1}$The University of Arizona, Tucson, USA}
      
      \vspace{0.2cm}
      
      \textbf{$^{2}$Center for Theoretical Physics, Polish Academy of Science, Warsaw, Poland}

 	\vspace{0.2cm}
      
      \textbf{$^{3}$ La Universidad CEU Cardenal Herrera ,  Valencia, Spain}

    \vspace{0.9cm}
    \textbf{Abstract:}
In this paper, we contribute to the mathematical foundations of the recently established theory of Spectrum Broadcast Structures (SBS). These are multipartite quantum states, encoding an operational notion of objectivity and exhibiting a more advanced form of decoherence. We study SBS and asymptotic convergence to SBS in the case of a central system interacting with $N$ environments via the von Neumann-type measurement interactions, ubiquitous in the theory of open quantum systems. We will be focusing on the case where the system is modeled by an infinite-dimensional Hilbert space and the operators associated with the system in the Hamiltonian have purely continuous spectrum. Such a setup yields mathematical complications that have hitherto not been addressed in the theory of SBS.
\end{center}
\onehalfspacing

\section{Background and Fundamental Concepts}
\;\;\; In recent times significant attention has been given to a family of multipartite states named \emph{Spectrum Broadcast Structures} (SBS) \cite{JKone} \cite{JKtwo} \cite{JKthree} \cite{kor5}. Since its creation, the theory of SBS has been used as a tool in the discipline of \emph{quantum foundations}; particularly in the theories of \emph{quantum decoherence} and  \emph{quantum darwinism}\cite{schloss}\cite{schloss2}\cite{Zurek}\cite{Zurek2}. Recently, quantum darwinism and SBS theory have been shown to be equivalent under certain technical assumptions \cite{Olaya}. Motivating the theory of quantum darwinism and the theory of SBS is the question of objectivity in the quantum world. To avoid philosophical contention, \cite{JKone} \cite{JKtwo} and\cite{JKthree} provide a definition of objectivity motivated by properties of classical dynamical systems. A multipartite quantum mechanical state satisfying such properties is called a SBS. The definition of objectivity proposed in \cite{Zurek} is:

\begin{definition}
\label{def:objectivity} A state of the system $S$ exists objectively if many observers can find out the state of $S$ independently, and without perturbing it.
\end{definition}

There are two clauses in the definition above that need to be made more precise, namely, "can find out the state of $S$" and "without perturbing it". The first of these means that any of the observers may locally solve a \emph{quantum state discrimination} optimization problem (QSD) \cite{bae} \cite{qiu} \cite{hellstrom} that allows the observer to identify the state of the system $S$ by proxy; we include a brief discussion regarding QSD in appendix \ref{app:QSD}. The second clause, "without perturbing it" may be formalized by introducing a distance measure. In this work, we will be using the trace distance, associated with the norm $\|\mathbf{\hat{A}}\big\|_{1}:=\sqrt{\mathbf{\hat{A}}^{\dagger}\mathbf{\hat{A}}}$.  The following definition proposed in \cite{JKone} is a mathematical formalization of Definition \ref{def:objectivity}. 
\begin{definition}{\textbf{SBS}:}
\label{eqn:turk4}
A Spectrum Broadcast Structure is a multipartite state of a central system $S$ and an environment 
 $E$, consisting of sub-environments $E^{1}, E^{2},..., E^{N_{E}}$:
\begin{equation}
\label{eqn:morocco}
\boldsymbol{\hat{\rho}} = \sum_{i}p_{i}|i\rangle\langle i|\otimes \bigotimes_{k= 1}^{N_{E}}\boldsymbol{\hat{\rho}}_{i}^{E^{k}}
\end{equation}
where $\{|i\rangle\}_{i}$ is some basis in the system's space, $p_{i}$ are probabilities summing up to one, and all states $\boldsymbol{\hat{\rho}}_{i}^{E^{k}}$ are perfectly distinguishable in the following sense:
\begin{equation}
F\big(\boldsymbol{\hat{\rho}}_{i}^{E^{k}},\boldsymbol{\hat{\rho}}_{j}^{E^{k}}\big) =  0
\end{equation}
for all $i\neq j$ and for all $k = 1,...,N_{E}$, where $F(\cdot,\cdot)$ is the quantum fidelity defined as $F\big(\boldsymbol{\hat{\rho}},\boldsymbol{\hat{\sigma}}\big):=\big\|\sqrt{\boldsymbol{\hat{\rho}}}\sqrt{\boldsymbol{\hat{\sigma}}}\big\|_{1}^{2}$ \cite{fuchs}.
\end{definition}
\subsection{Non-unitary dynamics from von Neumann type interactions in the case of the discrete variable}
\;\;\; In this paper, we will study non-unitary dynamics in the \emph{quantum-measurement} limit (this is the name given to the limit where the interaction Hamiltonian dominates the dynamics, see \cite{schloss}) generated in  by von Neumann-type interaction Hamiltonians \cite{zeh}, namely those of the form
\begin{equation}
\label{end:dynamicss}
\mathbf{\hat{H}}_{int} = \mathbf{\hat{X}}\otimes\sum_{k}g_{k}\mathbf{\hat{B}}_{k}
\end{equation}
where $\mathbf{\hat{X}}$ and the $\mathbf{\hat{B}}_{k}$ are self-adjoint operators acting in their respective Hilbert spaces $\mathscr{H}_{S}$ and $\mathscr{H}_{E^{k}}$ for each $k$, and $g_k$ are (real) coupling constants. We study a quantum system interacting with $N$ macroscopic environments, and assume that the joint initial state has the product form: 
\begin{equation}
\label{eqn:snake}
\boldsymbol{\hat{\rho}}=\boldsymbol{\hat{\rho}}_{S_{0}}\otimes\bigotimes_{k=1}^{N}\boldsymbol{\hat{\rho}}^{E^{k},0}
\end{equation}
In the state (\ref{eqn:snake}) we write the subscript $0$ of $\hat{\boldsymbol{\rho}}$ followed by a comma, i.e. $E^{k},0$, to emphasize that this is the initial state of the $k$th environment $E^{k}$; similarly we use the subscript in $S_{0}$ to highlight the initial state of the system. We evolve our total initial state using the dynamics generated by (\ref{end:dynamicss}), i.e  
\begin{equation}
\label{eqn:zelda10}
\boldsymbol{\hat{\rho}}_{t} = \bigg(e^{-it\mathbf{\hat{X}}\otimes\sum_{k=1}^{N}g_{k}\mathbf{\hat{B}}_{k}}\bigg)\boldsymbol{\hat{\rho}}_{S_{0}}\otimes\bigotimes_{k=1}^{N}\boldsymbol{\hat{\rho}}^{E^{k},0}\bigg(e^{it\mathbf{\hat{X}}\otimes\sum_{k=1}^{N}g_{k}\mathbf{\hat{B}}_{k}}\bigg).
\end{equation}
To study the state of the subsystem formed by the system $S$ and the first $N_{E}$ environments, we take the partial trace of the time-evolved density operator over the remaining $M_{E}:=N-N_{E}$ environments. For the case where $\mathbf{dim}\big(\mathscr{H}_{S}\big) = d_{S}<\infty $ is a finite-dimensional Hilbert space and the operator $\mathbf{\hat{X}}$ has full rank, the result is 
\begin{equation}
\label{eqn:decmode}
\sum_{i,j=1}^{d_{S}}\sigma_{i,j}\Gamma(i,j,t)|i\rangle\langle j| \otimes\bigotimes_{k=1}^{N_{E}} \boldsymbol{\hat{\rho}}_{x_{i},x_{j}}^{E^{k},t}
\end{equation}
where, $\{\big|i\big\rangle\}_{i=1}^{d_{s}}$ are the eigenvectors of $\mathbf{\hat{X}}$, with corresponding eigenvalues $\{x_{i}\}_{i=1}^{d_{S}}$ and we have used the following notation
\begin{equation}
\boldsymbol{\hat{\rho}}_{x,y}^{E^{k},t}:= e^{-itxg_{k}\mathbf{\hat{B}}_{k}}\boldsymbol{\hat{\rho}}^{E^{k},0}e^{ityg_{k}\mathbf{\hat{B}}_{k}} \: (k  =1,2,..., N_{E}) 
\end{equation}
\begin{equation}
\label{eqn:purepreserve}
\boldsymbol{\hat{\rho}}_{x}^{E^{k},t}:= e^{-itxg_{k}\mathbf{\hat{B}}_{k}}\boldsymbol{\hat{\rho}}^{E^{k},0}e^{itxg_{k}\mathbf{\hat{B}}_{k}} \: (k = 1,2,...,N_{E}). 
\end{equation}
\begin{equation}
\sigma_{i,j} := \langle i|\boldsymbol{\hat{\rho}}_{S_{0}}|j\rangle 
\end{equation}
\begin{equation}
\label{thegamm}
\Gamma(i,j,,t):=\prod_{n=N_{E}+1}^{N}\gamma_{i,j}^{n}(t)
\end{equation}
where
\begin{equation}
\gamma_{i,j}^{k}(t):= Tr\big\{ \boldsymbol{\hat{\rho}}_{x_{i},x_{j}}^{E^{k},t} \big\}
\end{equation}
We may also write (\ref{eqn:decmode}) as 
\begin{equation}
\label{eqn:decmodechan}
\Lambda_{t}\Big(\boldsymbol{\hat{\rho}}_{S_{0}}\otimes\bigotimes_{k=1}^{N_{E}}\boldsymbol{\hat{\rho}}^{E^{k},0}\Big):=
\sum_{i,j=1}^{d_{S}}\sigma_{i,j}\Gamma(i,j,t)|i\rangle\langle j| \otimes\bigotimes_{k=1}^{N_{E}} \boldsymbol{\hat{\rho}}_{x_{i},x_{j}}^{E^{k},t}
\end{equation}
where $\Lambda_{t}$ is a quantum channel (see \cite{Nielsen} for a discussion on quantum channels) defined as follows.
\begin{equation}
\label{eqn:morocco7}
\Lambda_{t}\big(\boldsymbol{\hat{\rho}}\big):=\mathscr{U}_{t}\circ\mathscr{E}_{t}\big(\boldsymbol{\hat{\rho}}\big)
\end{equation}
where 
\begin{equation}
\mathscr{U}_{t}\big(\mathbf{\hat{A}}\big):= e^{-it\mathbf{\hat{X}}\otimes  \sum_{k=1}^{N_{E}}g_{k}\mathbf{\hat{B}}_{k}}\big( \mathbf{\hat{A}}\big)e^{it\mathbf{\hat{X}}\otimes\sum_{k=1}^{N_{E}}g_{k}\mathbf{\hat{B}}_{k}}
\end{equation}
acts non-trivially in $\mathcal{S}\big(\mathscr{H}_{S}\otimes \bigotimes_{k=1}^{N_{E}}\mathscr{H}_{k}\big)$ ( $\mathcal{S}(\mathscr{H})$ is the space of density operators acting in $\mathscr{H}$) 
and  $\mathscr{E}_{t}$ acts non-trivially only in $\mathcal{S}\big(\mathscr{H}_{S})$ as follows. 
\begin{equation}
\label{eqn:decfinitecase}
\mathscr{E}_{t}(\mathbf{\hat{C}}):= \sum_{i,j=1}^{d_{S}}\langle i|\mathbf{\hat{C}}|j\rangle\Gamma(i,j,t)|i\rangle\langle j|
\end{equation}
In words, the trace-preserving quantum map $\Lambda_{t}$ is a composition of two trace-preserving quantum maps $\mathscr{U}_{t}$ and $\mathscr{E}_{t}$: a unitary map acting on $S$ and the environmental degrees of freedom that were not traced out and a non-unitary map acting locally in $S$.\\

Refocusing our attention to (\ref{eqn:decmode}), the question of whether or not such a joint state of the system and its $N_{E}$ environments converges to an SBS state, becomes the mathematical question of whether or not the (\ref{eqn:decmode}) gets arbitrarily close to an SBS state as $t$ becomes large, where closeness is quantified by the trace norm. In \cite{JKone} \cite{JKthree} and \cite{kor5} this question has been treated in the case where $\mathscr{H}_{S}$ is finite-dimensional and $\mathbf{\hat{X}}$ is full rank. So far, there has been no attempt to rigorously tackle the scenario where $\mathscr{H}_{S}$ is infinite-dimensional and $\mathbf{\hat{X}}$ is a self-adjoint operator with continuous spectrum.

\section{SBS for Continuous Variables}
\label{eqn:ptrace}
\;\;\; In this section, we will generalize the concept of SBS to include the case where the operator $\mathbf{\hat{X}}$ is a position operator acting in $\mathscr{H}_{S}= L^{2}\big( \mathbb{R} \big)$. To adapt SBS theory to such a case, we will need to study the decoherence of superpositions of the generalized eigenstates of $\mathbf{\hat{X}}$. In particular, for a state $\boldsymbol{\hat{\rho}}_{S_{0}}$ in $\mathcal{S}\big(L^{2}\big( \mathbb{R} \big)\big)$, and for any $x$ and $x'$ thought of as elements of the spectrum of $\hat{\mathbf{X}}$, we will study the behaviour of  
\begin{equation}
\label{eqn:coherenceterms}
\big\langle x^\prime\big|\mathscr{E}_{t}\big(\boldsymbol{\hat{\rho}}_{S}\big)\big|x\big\rangle 
\end{equation}
as a function of $t$ , where $\mathscr{E}_{t}$ is a quantum map that produces non-unitary dynamics, including dissipation and decoherence effects \cite{schloss}. The terms (\ref{eqn:coherenceterms}), also known as coherences, should be compared to analogous expressions (\ref{eqn:decfinitecase}) studied in the previous section, in which $\hat{\mathbf{C}}$ was a density operator decomposed in the basis of a finite rank operator.  The number of coherence terms was in that case equal $d_{S}(d_{S}-1)$. We henceforth refer to the case where $\mathbf{\hat{X}}$ has purely continuous spectrum as the SBS theory for continuous variables (SBSCV).\\

As in the previous section, we assume \emph{quantum-measurement limit}:
\begin{equation}
\mathbf{\hat{H}}_{tot} \approx \mathbf{\hat{H}}_{int}
\end{equation}
We will also assume an interaction Hamiltonian of the von Neumann type. Hence,
\begin{equation}
\label{eqn:intham2}
\mathbf{\hat{H}}_{int} = \mathbf{\hat{X}}\otimes\sum_{k=1}^{N}g_{k}\mathbf{\hat{B}}_{k}
\end{equation}
Here $\mathbf{\hat{X}}$ will be taken to be the position operator. An example, which we will explore in-depth, is the case where all of the $\mathbf{\hat{B}}_{k}$ are either position or momentum operators. The time evolution operator corresponding to (\ref{eqn:intham2}) is:
\begin{equation}
\label{eqn:timeev1}
\mathbf{\hat{U}}_{t} = e^{-it\mathbf{\hat{X}}\otimes\sum_{k=1}^{N}g_{k}\mathbf{\hat{B}}_{k}}.
\end{equation}
Considering a product state, acting on the appropriate tensor product Hilbert space, as our initial state, as we did in (\ref{eqn:snake}), we apply the time evolution operator (\ref{eqn:timeev1}). 
\begin{equation}
\label{end:toky}
\boldsymbol{\hat{\rho}}_{t} = \bigg(e^{-it\mathbf{\hat{X}}\otimes\sum_{k=1}^{N}g_{k}\mathbf{\hat{B}}_{k}}\bigg)\bigg(\boldsymbol{\hat{\rho}}_{S_{0}}\otimes \bigotimes_{k=1}^{N}\boldsymbol{\hat{\rho}}^{E^{k},0}\bigg)\bigg(e^{it\mathbf{\hat{X}}\otimes\sum_{k=1}^{N}g_{k}\mathbf{\hat{B}}_{k}}\bigg).
\end{equation}

In order to study the state of the subsystem formed by the system $S$ and the first $N_{E}$ environments, we take the partial trace of the time-evolved density operator over the remaining $M_{E}:=N-N_{E}$ environments. Using Lemma \ref{eqn:partialtrace} in appendix \ref{thelemmainthispape} and (\ref{end:toky}), the result of partially tracing (\ref{end:toky}) over $M_{E}$ environments is
\begin{equation}
\label{eqn:coca4}
\mathscr{U}_{N_{E},t}\Bigg(\mathscr{E}_{t}^{M_{E}}\big(\boldsymbol{\hat{\rho}}_{S_{0}}\big)\otimes\bigotimes_{k=1}^{N_{E}}\boldsymbol{\hat{\rho}}^{E^{k},0}\Bigg)
\end{equation}
where 
\begin{equation}
 \mathscr{U}_{n,t}\big( \mathbf{\hat{A}}\big) : = e^{-it\mathbf{\hat{X}}\otimes  \hat{\mathbf{S}}_{n}}\big( \mathbf{\hat{A}}\big)e^{it\mathbf{\hat{X}}\otimes\hat{\mathbf{S}}_{n}} 
\end{equation}
\begin{equation}
\hat{\mathbf{S}}_{n} := \sum_{k=1}^{n}g_{k}\mathbf{\hat{B}}_{k} 
\end{equation}
and
\begin{equation}
\label{eqn:thequantummap}
 \mathscr{E}_{t}^{M_{E}}\{ \boldsymbol{\hat{\sigma}}\} : =\int\int \langle x|\boldsymbol{\hat{\sigma}}|y\rangle\Gamma_{M_{E}}(t,x,y)|x\rangle\langle y| dxdy.
\end{equation}
where
\begin{equation}
\label{eqn:scrmyo}
\Gamma_{M_{E}}(t,x,y):=  \prod_{k=N_{E}+1}^{N}Tr_{k}\bigg\{\big(e^{-itxg_{k}\mathbf{\hat{B}}_{k}}\big) \boldsymbol{\hat{\rho}}^{E^{k},0}\big(e^{it yg_{k}\mathbf{\hat{B}}_{k}}\big)\bigg\} = 
\end{equation}
\begin{equation}
\prod_{k=N_{E}+1}^{N}Tr_{k}\bigg\{\big(e^{-it(x-y)g_{k}\mathbf{\hat{B}}_{k}}\big) \boldsymbol{\hat{\rho}}^{E^{k},0}\bigg\}   
\end{equation}
with $M_{E}= N-N_{E}$, the number of traces being taken in equation (\ref{eqn:scrmyo}). Note that $\gamma_{M_{E}}(t,x,y)=\gamma_{M_{E}}(t,x-y)$ is a function of $x-y$.
To simplify the notation for now, we shall forgo all but two environments, i.e. $N=2$, $N_{E} = M_{E}= 1$ (generalization to arbitrary $M_{E}$, $N_{E}$ and $N$ is straightforward). 
In this case, after partial
tracing over one of the environments we obtain the following density operator. 
\begin{equation}
\label{eqn:evo12}
\boldsymbol{\hat{\rho}}_{t}: = \mathscr{U}_{1,t}\big(\mathscr{E}^{1}_{t}\{\boldsymbol{\hat{\rho}}_{S_{0}}\}\otimes\boldsymbol{\hat{\rho}}^{E^{1},0}\big)
\end{equation}
The map $\mathscr{E}^{1}_{t}$ is a decoherence quantum map and $\mathscr{U}_{1,t}$ is a unitary map obtained from the Hamiltonian (\ref{eqn:intham2}) for the case $ N = 2$ and $N_{E}=1$ (here, $n=1$).\\

The primary divergence from the techniques presented in the previous section is that we now need to partition the operator $\mathscr{E}^{1}_{t}\{\boldsymbol{\hat{\rho}}_{S_{0}}\}$. For the case where $\mathbf{\hat{X}}$ is a position operator, the partitions of interest are of the form: 
\begin{equation}
\label{eqn:partitionuzbek}
\mathscr{E}^{1}_{t}\{\boldsymbol{\hat{\rho}}_{S_{0}}\big\} =\sum_{i}\sum_{j} \mathbf{\hat{P}}_{\Delta_{i,t}}\mathscr{E}^{1}_{t}\{\boldsymbol{\hat{\rho}}_{S_{0}}\big\}\mathbf{\hat{P}}_{\Delta_{j,t}}
\end{equation}
where the operators $\mathbf{\hat{P}}_{\Delta_{i,t}}$ are projector operators defined as follows.  $\mathbf{\hat{P}}_{\Delta_{i,t}} = \chi_{\Delta_{i,t}}\big(\mathbf{\hat{X}}\big)$ ($\chi_{\Delta_{i,t}}\big(x)$ are indicator functions and the $\Delta_{i,t}$ are subsets of the real line). This is akin to what was done in the previous section to obtain (\ref{eqn:decmode}), where in lieu of the projectors $\mathbf{\hat{P}}_{\Delta_{i,t}}$, projectors onto the eigensubspspaces corresponding to the eigenvectors $\big\{\big|i\big\rangle\big\}_{i=1}^{d_{S}}$ are used (see discussion following (\ref{eqn:decmode})). Although using projectors onto the generalized eigensubspaces of the position operator $\mathbf{\hat{X}}$ is the most natural way of generalizing (\ref{eqn:decmode}), it has  limitations. The projectors $\mathbf{\hat{P}}_{\Delta_{i,t}}$ will be in general time-dependent, and the size of the $\Delta_{i,t}$ will be restricted. We will not explore the quantum-metrological aspects of SBSCV in this work; we simply highlight the fact that the size of the magnitudes of the $\Delta_{i,t}$ may be bounded from below 
and from above by parameters depending on by quantum-metrological and other physical limitations \cite{Wiseman}.  A physical interpretation of the limiting smallness of the subspaces $\Delta_{i,t,}$ may be deduced within the context of von Neumann's theory of quantum measurement (see chapter 3 of \cite{Nielsen}). The set of $\mathbf{\hat{P}}_{\Delta_{i,t}}$ is a PVM (projector valued measure) and therefore characterizes a von Neumann measurement on the system $S$. With the latter in mind, the sizes of the $\Delta_{i,t}$ may be interpreted as resolution limits. Indeed, resolving the position of an arbitrarily small particle would require arbitrarily larger amounts of energy as the size of the particle becomes smaller. Due to the technological limitations of monitoring apparatuses, there will always be a limit to the smallness of the resolution $\Delta_{i,t}$. When introducing the approximate SBS state for continuous variables, a specific PVM acting on the system $S$ will be assumed for every $t$ prior to estimating the respective optimization problem that ensues (see (\ref{eqn:uzbek})). It is there where the partition (\ref{eqn:partitionuzbek}) will play a key role. \\

Assuming that we have an appropriate partition (\ref{eqn:partitionuzbek}), we may now mirror our work from the previous section in order to define an appropriate SBS for the CV case. We then develop tools to study the convergence of (\ref{eqn:partitionuzbek}) to such an SBS state in $t$. First, we present a definition that generalizes Definition \ref{eqn:turk4} to the CV setting. 

\begin{definition}[\textbf{SBS, a more general definition}]
\label{eqn:turkcv}
Let $\mathscr{H}_{S}\otimes\bigotimes_{k=1}^{N_{E}}\mathscr{H}_{E^{k}}$ be a tensor product of Hilbert spaces, with $\mathscr{H}_{S}$ corresponding to the system, and $\mathscr{H}_{E^{1}}, \mathscr{H}_{E^{2}},..., \mathscr{H}_{E^{N_{E}}}$---to environmental degrees of freedom. An SBS state acting in  $\mathcal{S}\big( \mathscr{H}_{S}\otimes\bigotimes_{k=1}^{N_{E}}\mathscr{H}_{E^{k}}
\big)$ is a density operator of the form
\begin{equation}
\boldsymbol{\hat{\rho}}_{SBS}:=\sum_{i}\Bigg(\mathbf{\hat{P}}_{S_{i}}\otimes\bigotimes_{k=1}^{N_{E}}\mathbb{I}_{E^{k}}\Bigg) \boldsymbol{\hat{\rho}}\Bigg(\mathbf{\hat{P}}_{S_{i}}\otimes\bigotimes_{k=1}^{N_{E}}\mathbb{I}_{E^{k}}\Bigg)
\end{equation}
satisfying the following properties (in the following, $\hat{\mathbf{P}}_{S_{i}}$ are projection operators.\\

Property 1) 

There exist projectors $\big\{\mathbf{\hat{P}}_{S_{i}}\big\}_{i}$ such that 
\begin{equation}
F\big(\boldsymbol{\hat{\rho}}_{i}^{E^{k}}, \boldsymbol{\hat{\rho}}_{j}^{E^{k}}\big) = 0 \;\; \forall \;\; i\neq j
\end{equation}
where $\big\{\mathbf{\hat{P}}_{S_{i}}\big\}$ is a PVM acting in $\mathscr{H}_{S}$ and
\begin{equation}
\boldsymbol{\hat{\rho}}_{i}^{E^{k}}:= T_{S}\Bigg\{T_{ E_{k^{'}\neq k}}\Bigg\{\Bigg(\mathbf{\hat{P}}_{S_{i}}\otimes\bigotimes_{k=1}^{N_{E}}\mathbb{I}_{E^{k}}\Bigg) \boldsymbol{\hat{\rho}}\Bigg(\mathbf{\hat{P}}_{S_{i}}\otimes\bigotimes_{k=1}^{N_{E}}\mathbb{I}_{E^{k}}\Bigg)\Bigg\}\Bigg\}
\end{equation}
($Tr_{E^{k^{\prime}\neq k}}$ means that we trace over all environments with the exception of the $k$th environment).\\

Property 2) $Tr_{S}\big\{\boldsymbol{\hat{\rho}}\big\}$ is a separable state. i.e. it is of the form 
\begin{equation}
Tr_{S}\big\{\boldsymbol{\hat{\rho}}\big\} = \sum_{i}p_{i}\bigotimes_{k=1}^{N_{E}} \boldsymbol{\hat{\rho}}^{E^{k}}_{i} \;\; \Big(\sum_{i}p_{i} = 1\Big)
\end{equation}
or
\begin{equation}
Tr_{S}\big\{\boldsymbol{\hat{\rho}}\big\} = \int p(x)\bigotimes_{k=1}^{N_{E}} \boldsymbol{\hat{\rho}}^{E^{k}}_{x} dx \;\; \Big(\int p(x) dx = 1\Big)
\end{equation}
\end{definition}

\vspace{4mm}

\;\;\;Indeed, it can be shown that a system, satisfying Definition \ref{eqn:turk4} satisfies also the properties of Definition \ref{eqn:turkcv}. Furthermore, for every $t>0$ one may deduce an approximate SBS state, in the sense of Definition \ref{eqn:turkcv} from $\boldsymbol{\hat{\rho}}_{t}$ (i.e. from equation (\ref{eqn:evo12})) as follows. Let the PVM $\big\{ \mathbf{\hat{P}}_{\Delta_{i,t}} \big\}_{i}$ and $\big\{\mathbf{\hat{P}}_{i}^{E^{1},t}   \big\}_{i}$ be PVM characterizing von Neumann measurements for the system $S$ and the environment $E^{1}$ respectively. We may use the PVM acting on $S$ for generating the partition (\ref{eqn:partitionuzbek}) as the PVM associated with the von Neumann measurement performed on $S$ at time $t$. The state
\begin{equation}
\label{eqn:SBSCVuzbek}
\frac{1}{\mathscr{N}(t)}\sum_{i}\big(\mathbf{\hat{P}}_{\Delta_{i,t}}\otimes\mathbf{\hat{P}}_{i}^{E^{1},t}\big)\boldsymbol{\hat{\rho}}_{t}\big(\mathbf{\hat{P}}_{\Delta_{i,t}}\otimes\mathbf{\hat{P}}_{i}^{E^{1},t}\big)
\end{equation}
where $\mathscr{N}(t)$ is the normalization constant equal to the trace of $\sum_{i}\big(\mathbf{\hat{P}}_{\Delta_{i,t}}\otimes\mathbf{\hat{P}}_{i}^{E^{1},t}\big)\boldsymbol{\hat{\rho}}_{t}\big(\mathbf{\hat{P}}_{\Delta_{i,t}}\otimes\mathbf{\hat{P}}_{i}^{E^{1},t}\big)$, is an SBS state which approximates $\boldsymbol{\hat{\rho}}_{t}$ at time $t$.\\ 

To get the SBS state, constructed via the algorithm described in the previous paragraph, closest (in the trace distance sense) to $\boldsymbol{\hat{\rho}}_{t}$ for a fixed $t>0$ we must solve the optimization problem 
\begin{equation}
\label{eqn:uzbek}
\min_{PVM}\frac{1}{2}\bigg\|\boldsymbol{\hat{\rho}}_{t} -  \frac{1}{\mathscr{N}(t)}\sum_{i}\big(\mathbf{\hat{P}}_{\Delta_{i,t}}\otimes\mathbf{\hat{P}}_{i}^{E^{1},t}\big)\boldsymbol{\hat{\rho}}_{t}\big(\mathbf{\hat{P}}_{\Delta_{i,t}}\otimes\mathbf{\hat{P}}_{i}^{E^{1},t}\big) \bigg\|_{1}
\end{equation}
where the minimum is taken over all PVM acting on the environmental degree of freedom. In general, it will not be possible to solve the problem exactly;  rather, we will estimate the minimum (\ref{eqn:uzbek}) by a more managable expression. To this end, we bound (\ref{eqn:uzbek}) as follows. \\
\begin{equation}
\label{eqn:uzbek5}
\min_{PVM}\frac{1}{2}\bigg\|\boldsymbol{\hat{\rho}}_{t} -  \frac{1}{\mathscr{N}(t)}\sum_{i}\big(\mathbf{\hat{P}}_{\Delta_{i,t}}\otimes\mathbf{\hat{P}}_{i}^{E^{1},t}\big)\boldsymbol{\hat{\rho}}_{t}\big(\mathbf{\hat{P}}_{\Delta_{i,t}}\otimes\mathbf{\hat{P}}_{i}^{E^{1},t}\big) \bigg\|_{1} =
\end{equation}
\begin{equation}
\label{eqn:uzbek6}
\min_{PVM}\frac{1}{2}\bigg\|\sum_{i}\sum_{j}\mathbf{\hat{P}}_{\Delta_{i,t}}\boldsymbol{\hat{\rho}}_{t}\mathbf{\hat{P}}_{\Delta_{j,t}} -  \frac{1}{\mathscr{N}(t)}\sum_{i}\big(\mathbf{\hat{P}}_{\Delta_{i,t}}\otimes\mathbf{\hat{P}}_{i}^{E^{1},t}\big)\boldsymbol{\hat{\rho}}_{t}\big(\mathbf{\hat{P}}_{\Delta_{i,t}}\otimes\mathbf{\hat{P}}_{i}^{E^{1},t}\big) \bigg\|_{1}
\end{equation}
which in turn may be bounded by a sum of two terms as follows.  
\begin{equation}
\label{eqn:uzbek7}
\min_{PVM}\frac{1}{2}\bigg\|\sum_{i}\sum_{j}\mathbf{\hat{P}}_{\Delta_{i,t}}\boldsymbol{\hat{\rho}}_{t}\mathbf{\hat{P}}_{\Delta_{j,t}} -  \frac{1}{\mathscr{N}(t)}\sum_{i}\big(\mathbf{\hat{P}}_{\Delta_{i,t}}\otimes\mathbf{\hat{P}}_{i}^{E^{1},t}\big)\boldsymbol{\hat{\rho}}_{t}\big(\mathbf{\hat{P}}_{\Delta_{i,t}}\otimes\mathbf{\hat{P}}_{i}^{E^{1},t}\big) \bigg\|_{1}\leq
\end{equation}

\begin{equation}
\label{eqn:uzbek8}
\min_{PVM}\Bigg(\frac{1}{2}\bigg\|\sum_{i}\mathbf{\hat{P}}_{\Delta_{i,t}}\boldsymbol{\hat{\rho}}_{t}\mathbf{\hat{P}}_{\Delta_{i,t}} -  \frac{1}{\mathscr{N}(t)}\sum_{i}\big(\mathbf{\hat{P}}_{\Delta_{i,t}}\otimes\mathbf{\hat{P}}_{i}^{E^{1},t}\big)\boldsymbol{\hat{\rho}}_{t}\big(\mathbf{\hat{P}}_{\Delta_{i,t}}\otimes\mathbf{\hat{P}}_{i}^{E^{1},t}\big) \bigg\|_{1}\Bigg)+
\end{equation}
\begin{equation}
\label{eqn:uzbek9}
\frac{1}{2}\bigg\|\sum_{i}\sum_{j;j\neq i}\mathbf{\hat{P}}_{\Delta_{i,t}}\boldsymbol{\hat{\rho}}_{t}\mathbf{\hat{P}}_{\Delta_{j,t}} \bigg\|_{1} 
\end{equation}
In what follows, we estimate the terms (\ref{eqn:uzbek8}) and (\ref{eqn:uzbek9}) separately, dedicating a separate section to each one of them. We will refer to the term (\ref{eqn:uzbek8}) as the \emph{diagonal term} and to the term (\ref{eqn:uzbek9}) as the \emph{off-diagonal term} (or the coherence term). We will first briefly comment on the main mathematical difficulty arising in SBSCV theory:  motivating the partition (\ref{eqn:partitionuzbek}). We will then present bounds, useful in studying the diagonal and off-diagonal terms (\ref{eqn:uzbek8}) and (\ref{eqn:uzbek9}) respectively.

\section{Problem with definition \ref{eqn:turk4} when Introducing Continuous Variables}
\;\;\; In this section, we shed light on our reasoning behind the new definition for SBS presented as Definition \ref{eqn:turkcv}.\\ 

To appreciate the difficulties, caused by the presence of continuous spectrum, let us examine the state (\ref{eqn:evo12}).  The system's state is now a density operator $\boldsymbol{\hat{\rho}}_{S_{0}}$ in an infinite-dimensional Hilbert space;  for our purposes, it will be convenient to take this space to be $L^{2}(\mathbb{R})$. For simplicity, analogously to (\ref{eqn:intham2}), we define the interaction of the system with the environment as
\begin{equation}
H_{int}= \gamma\mathbf{\hat{X}}\otimes\mathbf{\hat{B}}
\end{equation}
where $\mathbf{\hat{X}}$ is the position operator. Being a trace-class operator, $\boldsymbol{\hat{\rho}}_{S_{0}}$ can be represented as an integral operator with a kernel $K(x,y)$. 
The expansion analogous to (\ref{eqn:decmode}) is the following:
\begin{equation}
\label{eqn:stateinquestion39}
\boldsymbol{\hat{\rho}}_{t}= \int\int K(x,y)\gamma^{2}_{x,y}(t)|x\rangle\langle y|\otimes \boldsymbol{\hat{\rho}}^{E^{1},t}_{x,y}   dxdy
\end{equation}
where as before $\boldsymbol{\hat{\rho}}_{x,y}^{E^{1},t}:= e^{-ix\gamma\mathbf{\hat{B}}t}\boldsymbol{\hat{\rho}}^{E^{1},0}e^{iy\gamma\mathbf{\hat{B}}t}$ and $\gamma^{2}_{x,y}(t):= Tr\big\{ \boldsymbol{\hat{\rho}}_{x,y}^{E^{2},t} \big\}$. 
Unlike the state (\ref{eqn:decmode}), the state (\ref{eqn:stateinquestion39}) does not have a clear decomposition into diagonal and off-diagonal terms using the spectral decomposition of the operator $\boldsymbol{\hat{X}}$ in terms of generalized eigenvectors $\big|x\big\rangle$, which we have employed to expand $\mathscr{U}_{t}\big(\mathscr{E}_{t}\big(\boldsymbol{\hat{\rho}}_{S_{0}}\big)\otimes\boldsymbol{\hat{\rho}}^{E,0}\big) = \big(e^{-it\gamma\mathbf{\hat{X}}\otimes\mathbf{\hat{B}}}\big)\big(\mathscr{E}_{t}\big(\boldsymbol{\hat{\rho}}_{S_{0}}\big)\otimes\boldsymbol{\hat{\rho}}^{E,0}\big)\big(e^{-it\gamma\mathbf{\hat{X}}\otimes\mathbf{\hat{B}}}\big)$; herein we have dropped the subscripts and superscripts indicating that we have traced over one environmental degree of freedom, per the notional conventions prescribed in equations (\ref{eqn:coca4}) through (\ref{eqn:scrmyo}). We shall forgo usage of such superscripts and subscripts from now on, unless the values of $M_{E}$ and $N_{E}$ are relevant.\\

In the finite-dimensional case, we could clearly distinguish between diagonal and off-diagonal entries in order to deduce an SBS structure approximating the state in question (see section 4 of \cite{AA}). In the continuous variable case, taking the naive approach, this method breaks down since the diagonal term is now
\begin{equation}
\label{eqn:stateinquestion40}
\boldsymbol{\hat{\rho}}_{t}= \int  K(x,x)|x\rangle\langle x|\otimes \boldsymbol{\hat{\rho}}^{E,t}_{x}dx
\end{equation}
which is not a trace class operator, since it is unitarily equivalent to a tensor product of a multiplication operator and a trace class operator---thus it cannot represent a quantum state. Being able to separate between diagonal and off-diagonal terms in \cite{AA} was a key step in the estimations that led to Theorem 4 of \cite{AA}, a bound which estimates the analog or (\ref{eqn:uzbek}) for the case of discrete variables.  To proceed similarly for the CV case we must partition the state (\ref{eqn:stateinquestion39})by applying the partition (\ref{eqn:partitionuzbek}). \\

Another difficulty arising in the continuous variable case is an increase in complexity when dealing with trace norms; starting from the fact that $\big\|\big|x\big\rangle \big\langle y\big|\big\|_{1}$ is undefined for generalized states $\big|x\big\rangle$ and $\big|y\big\rangle$.

\section{Partitioning (\ref{eqn:evo12})}
\;\;\; To formally introduce our approach for the study of SBS in the CV case we will first discuss the phenomenon of decoherence in the context of the model (\ref{eqn:evo12}). In this case, decoherence arises from the quantum map $\mathscr{E}_{t}$ in:
\begin{equation}
\label{eqn:bazinga}
\boldsymbol{\hat{\rho}}_{t}  = \big(e^{-it\gamma\mathbf{\hat{X}}\otimes \mathbf{\hat{B}}}\big)\big(\mathscr{E}_{t}\big(\boldsymbol{\hat{\rho}}_{S_{0}}\big)\otimes \boldsymbol{\hat{\rho}}^{E,0}\big)\big(e^{it\gamma \mathbf{\hat{X}}\otimes \mathbf{\hat{B}}}\big)
\end{equation}
%

As we have done in (\ref{eqn:stateinquestion39}), we use the representation 
\begin{equation}
\label{eqn:represco}
\boldsymbol{\hat{\rho}}_{S_{0}} = \int\int K(x,y)|x\rangle \langle  y|dxdy
\end{equation}
using the generalized eigenvectors of the position operator $\mathbf{\hat{X}}$. Using representation (\ref{eqn:represco}), and refering back to (\ref{eqn:thequantummap}), the effect of $\mathscr{E}_{t}$ on $\boldsymbol{\hat{\rho}}_{S_{0}}$ is hence 
\begin{equation}
\label{eqn:contractivemap}
\mathscr{E}_{t}\big(\boldsymbol{\hat{\rho}}_{S_{0}}\big)  = \int\int K(x,y)\Gamma(t,x,y)\big|x\big\rangle \big\langle  y\big|dxdy
\end{equation}
where $\Gamma(t,x,y)$ is a kernel yielding non-unitary dynamics obtained via partial tracing as seen in (\ref{eqn:scrmyo}). Substituting this into (\ref{eqn:bazinga}) we obtain 
\begin{equation}
\label{eqn:contractexpand}
\boldsymbol{\hat{\rho}}_{t}= \int\int K(x,y)\Gamma(t,x,y)\big|x\big\rangle\big\langle y\big|\otimes \boldsymbol{\hat{\rho}}^{E,t}_{x,y}dxdy 
\end{equation}
where we remind the reader that $\boldsymbol{\hat{\rho}}^{E,t}_{x,y}:= e^{-it\gamma x\mathbf{\hat{B}}}\boldsymbol{\hat{\rho}}^{E,0} e^{it\gamma y\mathbf{\hat{B}}}$.\\


For fixed $t>0$, we adopt a partition characterized by a PVM $\mathbf{\hat{P}}_{\Delta_{i,t}}:=\chi_{\Delta_{i,t}}\big(\mathbf{\hat{X}}\big)$ acting on the degree of freedom pertaining to the system, as was done in (\ref{eqn:partitionuzbek}), in order to express (\ref{eqn:contractexpand}) as follows. 
\begin{equation}
\label{eqn:jarekcom}
\boldsymbol{\hat{\rho}}_{t}= \sum_{i}\sum_{j}\mathbf{\hat{P}}_{\Delta_{i,t}}\Big( \int\int K(x,y)\Gamma(t,x,y)\big|x\big\rangle\big\langle y\big|\otimes \boldsymbol{\hat{\rho}}^{E,t}_{x,y}dxdy   \Big)\mathbf{\hat{P}}_{\Delta_{j,t}} =
\end{equation}
\begin{equation}
\sum_{i}\sum_{j}\Big( \int\int K(x,y)\Gamma(t,x,y)\mathbf{\hat{P}}_{\Delta_{i,t}}\big|x\big\rangle\big\langle y\big|\mathbf{\hat{P}}_{\Delta_{j,t}} \otimes \boldsymbol{\hat{\rho}}^{E,t}_{x,y}dxdy\Big) =
\end{equation}
\begin{equation}
\label{eqn:contractiveelegant}
\sum_{i}\sum_{j}\int_{\Delta_{i,t}}\int_{\Delta_{j,t}} K(x,y)\Gamma(t,x,y)|x\rangle\langle y|\otimes \boldsymbol{\hat{\rho}}^{E,t}_{x,y} dxdy  
\end{equation}
Once again, we call the elements of the sum (\ref{eqn:contractiveelegant}) for which $i\neq j$ the "off-diagonal terms" and the terms for which $i=j$ "the diagonal terms".

\section{Estimating the "Off-diagonal" Terms (\ref{eqn:uzbek9})}
\;\;\; We now use the (\ref{eqn:contractiveelegant}) as a candidate to fulfil Definition \ref{eqn:turkcv} and proceed by estimating (\ref{eqn:uzbek9}). 
\begin{equation}
\label{eqn:the52}
  \bigg\|\sum_{i}\sum_{j;j\neq i}\int_{\Delta_{i,t}}\int_{\Delta_{j,t}}K(x,y)\Gamma(t,x,y)\big|x\big\rangle\big\langle y\big|\otimes \boldsymbol{\hat{\rho}}^{E,t}_{x,y}dxdy\bigg\|_{1} =    
\end{equation}
\begin{equation}
   \bigg\|e^{-it\gamma\mathbf{\hat{X}}\otimes \mathbf{\hat{B}}}\Bigg(\bigg(\sum_{i}\sum_{j;j\neq i}\int_{\Delta_{i,t}}\int_{\Delta_{j,t}}K(x,y)\Gamma(t,x,y)\big|x\big\rangle\big\langle y\big|dxdy \bigg)\otimes \boldsymbol{\hat{\rho}}^{E,0}\Bigg)e^{it\gamma \mathbf{\hat{X}}\otimes \mathbf{\hat{B}}}\bigg\|_{1}  =  
\end{equation}
\begin{equation}
   \bigg\|\Bigg(\sum_{i}\sum_{j;j\neq i}\int_{\Delta_{i,t}}\int_{\Delta_{j,t}}K(x,y)\Gamma(t,x,y)\big|x\big\rangle\big\langle y\big|dxdy \Bigg)\otimes \boldsymbol{\hat{\rho}}^{E,0}\bigg\|_{1}  = 
\end{equation}
\begin{equation}
   \bigg\|\sum_{i}\sum_{j;j\neq i}\int_{\Delta_{i,t}}\int_{\Delta_{j,t}} K(x,y)\Gamma(t,x,y)\big|x\big\rangle\big\langle y\big|dxdy \bigg\|_{1}\leq 
\end{equation}
\begin{equation}
   \sum_{i}\sum_{j;j\neq i}\bigg\|\int_{\Delta_{i,t}}\int_{\Delta_{j,t}}  K(x,y)\Gamma(t,x,y)\big|x\big\rangle\big\langle y\big|dxdy\bigg\|_{1} = 
\end{equation}
\begin{equation}
\label{eqn:kupschbound}
 \sum_{i}\sum_{j;j\neq i}\bigg\|\mathbf{\hat{P}}_{\Delta_{i,t}}\mathscr{E}_{t}\big(\boldsymbol{\hat{\rho}}_{S_{0}}\big)\mathbf{\hat{P}}_{\Delta_{j,t}}\bigg\|_{1} 
\end{equation}
where again $\mathbf{\hat{P}}_{\Delta_{i,t}}:= \chi_{\Delta_{i,t}}\big(\mathbf{\hat{X}}\big) = \int_{\Delta_{i,t}}\big|x\big\rangle\big\langle x\big| dx$, i.e. the spectral projector of $\mathbf{\hat{X}}$ projecting onto the subspace corresponding to the set $\Delta_{i,t}$. The trace norms $\big\|\mathbf{\hat{P}}_{\Delta_{i}}\mathscr{E}_{t}\big(\boldsymbol{\hat{\rho}}_{S_{0}}\big)\mathbf{\hat{P}}_{\Delta_{j}}\big\|_{1} $ are in general quite difficult to estimate. We present below two approaches; one is an adaptation of the work in \cite{kupsch} and the other is an application of the main theorem of \cite{stolz}. 
\subsection{Bounds of the Kupsch Kind \cite{kupsch}}
\label{eqn:sec4.1}
\;\;\; One approach to estimating the trace norms in the inequality (\ref{eqn:kupschbound2}) below invokes some ideas from Kupsch's seminal paper on decoherence \cite{kupsch}, where it is proven that  
\begin{equation}
\label{eqn:kupschbound2}
\|P_{\Delta_{i}}\mathscr{E}_{t}\big(\boldsymbol{\hat{\rho}}_{S_{0}}\big)P_{\Delta_{j}}\|\leq C(1+\delta^{2}\psi(t))^{-\gamma}
\end{equation}
for intervals $\Delta_{j}$ and $\Delta_{i}$ separated by a distance $\delta >0$.
Here, $\psi(t)\geq 0$ is a function that diverges for $t\rightarrow \infty$, $\gamma$ an exponent which can be large, and  $C$ is some constant.  We will use a modification of (\ref{eqn:kupschbound2}), not its exact form. 
Again, we will focus on the case where $\mathbf{\hat{X}}$ is the position operator. \\

\begin{definition4}[\textbf{Adapting Kupsch's Bounds from} \cite{kupsch} \textbf{Appendix A}]
\label{eqn:theoremkupsch}
Let us fix $t>0$ and let $\boldsymbol{\hat{\rho}}_{t}$ be a density operator which may be represented, using the generalized eigenvectors of the position operator $\mathbf{\hat{X}}$, as
\begin{equation}
\boldsymbol{\hat{\rho}}_{t} = \int\int  \Gamma(t,x,y)K(x,y)\big|x\big\rangle \big\langle y\big|dxdy 
\end{equation}
where $\Gamma(t,x,y)\in \mathcal{C}^{1}\big(\mathbb{R}^{3}\big)$ and the $\mathbf{\hat{P}}_{\Delta_{i,t}}$ are defined as they were in (\ref{eqn:partitionuzbek}); $K(x,y)$ is the kernel of a state $\boldsymbol{\hat{\rho}}_{0}$. Then,

\begin{equation}
\label{eqn:kupschkupsch}
\|\mathbf{\hat{P}}_{\Delta_{i,t}}\boldsymbol{\hat{\rho}}_{t}\mathbf{\hat{P}}_{\Delta_{j,t}}\|_{1}\leq \sup_{(x,y)\in \Delta_{i,t}\times \Delta_{j,t}}\bigg(2|\Gamma(t,x,y)|+|\Delta_{j,t}||\partial_{y}\Gamma(t,x,y)|\bigg)
\end{equation}
when
\begin{equation}
\Big|\Delta_{i,t}\times\Delta_{j,t}\cap \mathbf{supp}\big\{\Gamma(t,x,y)K(x,y)\big\}\Big|\neq 0
\end{equation}
otherwise 
\begin{equation}
\|\mathbf{\hat{P}}_{\Delta_{i,t}}\boldsymbol{\hat{\rho}}_{t}\mathbf{\hat{P}}_{\Delta_{j,t}}\|_{1} = 0
\end{equation}
\end{definition4}

\begin{proof}

\vspace{4mm}

CASE 1) \\

If $\Delta_{i,t}\times\Delta_{j,t}$ is such that 
\begin{equation}
\Big|\Delta_{i,t}\times\Delta_{j,t}\cap \mathbf{supp}\big\{\Gamma(t,x,y)K(x,y)\big\}\Big| =  0
\end{equation}
then
\begin{equation}
\bigg\|\mathbf{\hat{P}}_{\Delta_{i,t}}\boldsymbol{\hat{\rho}}_{t}\mathbf{\hat{P}}_{\Delta_{j,t}}\bigg\|_{1} = \bigg\|\int_{\Delta_{i,t}}\int_{\Delta_{j,t}}\Gamma(t,x,y)K(x,y)\big|x\big\rangle\big\langle y\big| dxdy\bigg\|_{1} = 
\end{equation}
\begin{equation}
\bigg\|\int_{\Delta_{i,t}}\int_{\Delta_{j,t}}0\big|x\big\rangle\big\langle y\big| dxdy\bigg\|_{1} = 0
\end{equation}

\vspace{4mm}

CASE 2)

Now we assume that 
\begin{equation}
\Big|\Delta_{i,t}\times\Delta_{j,t}\cap \mathbf{supp} \big\{\Gamma(t,x,y)K(x,y)\big\}\Big| \neq 0
\end{equation}
Let us begin by considering the operator 
 \begin{equation}
\mathbf{\hat{T}}_{t}(y):= \int_{\Delta_{i,t}} \Gamma(t,x,y)\big|x\big\rangle\big\langle x\big|dx
 \end{equation}
where $i$ is fixed. $\mathbf{\hat{T}}_{t}(y)$ is a family of operators, differentiable with respect to $y$,  satisfying the operator norm estimate
\begin{equation} 
 \big\|\mathbf{\hat{T}}_{t}(y)\big\|\leq \sup_{x\in\Delta_{i,t}}|\Gamma(t,x,y)|
 \end{equation}
as shown by the following estimate.
 \begin{equation}
\big\|\mathbf{\hat{T}}_{t}(y)\big\|^{2} = \sup_{\||\psi\rangle\|=1}\big\|\mathbf{\hat{T}}_{t}(y)\big|\psi\big\rangle\big\|^{2}=
\end{equation}
\begin{equation}
\sup_{\||\psi\rangle\|=1}\int_{\Delta_{i,t}}\int_{\Delta_{i,t}} \Gamma(t,x^\prime,y)^{*}\Gamma(t,x,y)\big\langle \psi\big|x^\prime\big\rangle\big\langle x^\prime\big|x\rangle\big\langle x\big|\psi\big\rangle dx^\prime dx=
\end{equation}
\begin{equation}
\sup_{\||\psi\rangle\|=1}\int_{\Delta_{i,t}} |\Gamma(t,x,y)|^{2}\big\langle \psi\big|x\big\rangle\big\langle x\big|\psi\big\rangle dx\leq 
\end{equation}
\begin{equation}
\sup_{x\in \Delta_{i,t}}|\Gamma(t,x,y)|^{2}\sup_{\||\psi\rangle\|=1}\int_{\Delta_{i,t}}|\psi(x)|^{2}dx \leq \sup_{x\in \Delta_{i,t}}|\Gamma(t,x,y)|^{2}
\end{equation}
\vspace{5mm}
In a similar way, we may bound the operator norm of $\mathbf{\hat{T}}^{\prime}_{t}(y):=\int_{\Delta_{i,t}}\Gamma^{'}(t,x,y)\big|x\big\rangle\big\langle x\big|dx$
where $\Gamma^{'}(t,x,y):=\partial_{y}\Gamma(t,x,y)$. i.e. 
\begin{equation}
\big\|\mathbf{\hat{T}}^{\prime}_{t}(y)\big\|\leq \sup_{x\in \Delta_{i,t}}|\Gamma^{'}(t,x,y)|
\end{equation}
Furthermore, define $\mathbf{\hat{J}}_{t}(y):= \mathbf{\hat{T}}_{t}(y)\boldsymbol{\hat{\rho}}_{0}$ and $\mathbf{\hat{J}}_{t}^{'}(y):=\mathbf{\hat{T}}_{t}^{'}(y)\boldsymbol{\hat{\rho}}_{0} $. Using the above operator norm estimates and the inequality $\|\mathbf{\hat{A}}\mathbf{\hat{C}}\|_{1}\leq \|\mathbf{\hat{A}}\|\|\mathbf{\hat{C}}\|_{1}$ we obtain
\begin{equation}
\big\|\mathbf{\hat{J}}_{t}(y)\big\|_{1}\leq \sup_{x\in \Delta_{i,t}}|\Gamma(t,x,y)|\big\|\boldsymbol{\hat{\rho}}_{0}\big\|_{1} = \sup_{x\in \Delta_{i,t}}|\Gamma(t,x,y)|
\end{equation}
 and 
\begin{equation}
\big\|\mathbf{\hat{J}}_{t}^{'}(y)\big\|_{1}\leq \sup_{x\in \Delta_{i,t}}|\Gamma^{'}(t,x,y)|\big\| \boldsymbol{\hat{\rho}}_{0} \big\|_{1}=  \sup_{x\in \Delta_{i,t}}|\Gamma^{'}(t,x,y)|.
\end{equation}
We now clarify the relationship between the operator $\mathbf{\hat{T}}_{t}^{\prime}(y)$ and the derivative $\partial_{y}\big\langle \psi\big|\mathbf{\hat{T}}_{t}(y)\big|\phi\big\rangle$.
\begin{equation}
\partial_{y}\big\langle \psi\big|\mathbf{\hat{T}}_{t}(y)\big|\phi\big\rangle = \partial_{y}\int_{\Delta_{i,t}}\Gamma(t,x,y)\big\langle \psi\big|x\big\rangle\big\langle x\big|\phi\big\rangle  dx
\end{equation} 
Assuming that $\Gamma(t,x,y)$ is $C^{1}(\Delta_{i,t})$ with respect to $y$ we may now swap the order of the integral and the derivative. 
\begin{equation}
\partial_{y}\int_{\Delta_{i,t}}\Gamma(t,x,y)\big\langle \psi\big|x\big\rangle\big\langle x\big|\phi\big\rangle dx = \int_{\Delta_{i,t}}\partial_{y}\Gamma(t,x,y)\big\langle \psi\big|x\big\rangle\big\langle x\big|\phi\big\rangle dx = 
\end{equation}
\begin{equation}
\int_{\Delta_{i,t}}\Gamma^{'}(t,x,y)\big\langle \psi\big|x\big\rangle\big\langle x\big|\phi\big\rangle  dx  = \big\langle \psi\big|\Big(\int_{\Delta_{i,t}}\Gamma^{'}(t,x,y)\big|x\big\rangle\big\langle x\big|dx\Big)\big|\phi\big\rangle = \big\langle \psi\big|\boldsymbol{\hat{T}}_{t}^{\prime}(y)\big|\phi\big\rangle
\end{equation}
We therefore have 
\begin{equation}
 \partial_{y}\big\langle \psi\big|\mathbf{\hat{T}}_{t}(y)\big|\phi\big\rangle =   \big\langle \psi\big|\mathbf{\hat{T}}_{t}^{\prime}(y)\big|\phi \big\rangle
\end{equation}
Now, for all intervals $\Delta_{j,t}$, we have $\int_{\Delta_{j,t}}\mathbf{\hat{J}}_{t}(y)\big|y\big\rangle\big\langle y\big|dy = \mathbf{\hat{P}}_{\Delta_{i,t}}\boldsymbol{\hat{\rho}}_{t}\mathbf{\hat{P}}_{\Delta_{j,t}}$. Let us define $\Delta_{j,t}:= [a_{j}(t),b_{j}(t)]$. We will show that the following identity holds. 
\begin{equation}
\label{eqn:IBPFS}
\int_{\Delta_{j,t}}\mathbf{\hat{J}}_{t}(y)\big|y\big\rangle\big\langle y\big|dy= \mathbf{\hat{J}}_{t}(b_{j}(t))\mathbf{\hat{P}}_{(-\infty,b_{j}(t)]}-\mathbf{\hat{J}}_{t}(a_{j}(t))\mathbf{\hat{P}}_{(-\infty,a_{j}(t)]} - \int_{\Delta_{j,t}}\mathbf{\hat{J}}_{t}^\prime(y)\mathbf{\hat{P}}_{(-\infty,y]}dy.
\end{equation}
The state $\boldsymbol{\hat{\rho}}_{0}$ may be generally written as a mixture $\boldsymbol{\hat{\rho}}_{0}= \sum_{n}p_{n}\big|\xi_{0}^{n}\big\rangle\big\langle\xi_{0}^{n}\big|$. Now, for arbitrary $\big|\psi\big\rangle$ and $\big|\phi\big\rangle$  
\begin{equation}
\label{eqn:abouttodointbypart}
\big\langle\psi\big|\int_{\Delta_{j,t}}\mathbf{\hat{J}}_{t}(y)\big|y\big\rangle\big\langle y\big|dy\big|\phi\big\rangle = \int_{\Delta_{j,t}}\big\langle \psi\big|\mathbf{\hat{J}}_{t}(y)\big|y\big\rangle\big\langle y\big|\phi\big\rangle dy.
\end{equation}
By the definition of $\mathbf{\hat{J}}_{t}(y)$ one has 
\begin{equation}
  \big\langle\psi\big|\mathbf{\hat{J}}_{t}(y)\big|y\big\rangle=\big\langle\psi\big|\mathbf{\hat{T}}_{t}(y)\boldsymbol{\hat{\rho}}_{0}\big|y\big\rangle = \big\langle\psi\big|\mathbf{\hat{T}}_{t}(y)\Big(\sum_{n}p_{n}\big|\xi_{0}^{n}\big\rangle\big\langle\xi_{0}^{n}\big|\Big)\big|y\big\rangle.
\end{equation}
The right-hand side of $(\ref{eqn:abouttodointbypart})$ equals
\begin{equation}
\sum_{n}p_{n}\int_{\Delta_{j,t}}\big\langle\psi\big|\mathbf{\hat{T}}_{t}(y)\big|\xi^{n}_{0}\big\rangle\big\langle \xi^{n}_{0}\big|y\big\rangle\big\langle y\big|\phi\big\rangle dy= 
\end{equation}

\begin{equation}
 \sum_{n}p_{n}\Bigg[\big\langle\psi|\mathbf{\hat{T}}_{t}(y)\big|\xi^{n}_{0}\big\rangle\big\langle\xi^{n}_{0}\big|\bigg(\int_{-\infty}^{y}dy^\prime\big|y^\prime\big\rangle\big\langle y^\prime\big|\bigg)\big|\phi\big\rangle \Bigg]\Bigg|_{a_{j}(t)}^{b_{j}(t)}- \sum_{n}p_{n}\int_{\Delta_{j,t}}\Big(\int_{-\infty}^{y}\big\langle \xi^{n}_{0}\big|y^\prime\big\rangle \big\langle y^\prime\big|\phi\big\rangle dy^\prime \Big) d\bigg(\big\langle\psi\big|\mathbf{\hat{T}}_{t}(y)\big|\xi^{n}_{0}\big\rangle\bigg) =  
\end{equation}
\begin{equation}
\sum_{n}p_{n}\Bigg[\big\langle\psi\big|\mathbf{\hat{T}}_{t}(y)\big|\xi^{n}_{0}\big\rangle\big\langle\xi^{n}_{0}\big|P_{(-\infty,y]}\big|\phi\big\rangle \Bigg]\Bigg|_{a_{j}(t)}^{b_{j}(t)}- \sum_{n}p_{n}\int_{\Delta_{j,t}}\Big(\int_{-\infty}^{y}\big\langle \xi^{n}_{0}\big|y^\prime\big\rangle \big\langle y^\prime\big|\phi\big\rangle dy^\prime \Big) \bigg(\big\langle\psi\big|\mathbf{\hat{T}}_{t}(y)\big|\xi^{n}_{0}\big\rangle\bigg)^\prime dy =     
\end{equation}
\begin{equation}
\sum_{n}p_{n}\Bigg[\big\langle\psi\big|\mathbf{\hat{T}}_{t}(y)\big|\xi^{n}_{0}\big\rangle\big\langle\xi^{n}_{0}\big|\mathbf{\hat{P}}_{(-\infty,y]}\big|\phi\big\rangle \Bigg]\Bigg|_{a_{j}(t)}^{b_{j}(t)}- \sum_{n}p_{n}\int_{\Delta_{j,t}}\bigg(\big\langle\psi\big|\mathbf{\hat{T}}_{t}(y)\big|\xi^{n}_{0}\big\rangle\bigg)^\prime\big\langle \xi^{n}_{0}\big|\mathbf{\hat{P}}_{(-\infty,y]}\big|\phi\big\rangle dy =  
\end{equation}
\begin{equation}
\sum_{n}p_{n}\big\langle\psi\big|\Bigg[\mathbf{\hat{T}}_{t}(y)\big|\xi^{n}_{0}\big\rangle\big\langle\xi^{n}_{0}\big|\mathbf{\hat{P}}_{(-\infty,y]}\Bigg]\Bigg|_{a_{j}(t)}^{b_{j}(t)}\big|\phi\big\rangle - \sum_{n}p_{n}\big\langle\psi\big|\Bigg(\int_{\Delta_{j,t}}\mathbf{\hat{T}}^{\prime}_{t}(y)\big|\xi^{n}_{0}\big\rangle\big\langle \xi^{n}_{0}\big|\mathbf{\hat{P}}_{(-\infty,y]} dy \Bigg)\big|\phi\big\rangle =    
\end{equation}

\begin{equation}
\big\langle\psi\big|\Bigg[\mathbf{\hat{T}}_{t}(y)\boldsymbol{\hat{\rho}}_{0}\mathbf{\hat{P}}_{(-\infty,y]}\Bigg]\Bigg|_{a_{j}(t)}^{b_{j}(t)}\big|\phi\big\rangle - \big\langle\psi\big|\Bigg(\int_{\Delta_{j,t}}\mathbf{\hat{T}}^{\prime}_{t}(y)\boldsymbol{\hat{\rho}}_{0}\mathbf{\hat{P}}_{(-\infty,y]} dy \Bigg)\big|\phi\big\rangle =    
\end{equation}

\begin{equation}
\big\langle\psi\big|\Bigg(\mathbf{\hat{J}}_{t}(y)\mathbf{\hat{P}}_{(-\infty,y]}\Bigg|_{a_{j}(t)}^{b_{j}(t)}- \int_{\Delta_{j,t}}\mathbf{\hat{J}}^{\prime}_{t}(y) \mathbf{\hat{P}}_{(-\infty,y]} dy\bigg)\big|\phi\big\rangle  =  
\end{equation}
\begin{equation}
\big\langle\psi\big|\Bigg(\mathbf{\hat{J}}_{t}(b_{j}(t))\mathbf{\hat{P}}_{(-\infty,b_{j}(t)]}-\hat{J}_{t}(a_{j}(t))\mathbf{\hat{P}}_{(-\infty,a_{j}(t)]}- \int_{\Delta_{j,t}}\mathbf{\hat{J}}_{t}^{\prime}(y) \mathbf{\hat{P}}_{(-\infty,y]} dy\Bigg)\big|\phi\big\rangle   
\end{equation}
and so 
\begin{equation}
\int_{\Delta_{j,t}}\mathbf{\hat{J}}_{t}(y)\big|y\big\rangle\big\langle y\big|dy= \mathbf{\hat{J}}_{t}(b_{j}(t))\mathbf{\hat{P}}_{(-\infty,b_{j}(t)]}-\mathbf{\hat{J}}_{t}(a_{j}(t))\mathbf{\hat{P}}_{(-\infty,a_{j}(t)]} - \int_{\Delta_{j,t}}\mathbf{\hat{J}}_{t}^\prime(y)\mathbf{\hat{P}}_{(-\infty,y]}dy.
\end{equation}

Consequently
\begin{equation}
\big\|\mathbf{\hat{P}}_{\Delta_{i.t}}\boldsymbol{\hat{\rho}}_{t}\mathbf{\hat{P}}_{\Delta_{j,t}}\big\|_{1} =\bigg\|\int_{\Delta_{j,t}}\mathbf{\hat{J}}_{t}(y)|y\rangle\langle y|dy\bigg\|_{1}=  
\end{equation}
\begin{equation}
\bigg\|\mathbf{\hat{J}}_{t}(b_{j}(t))\mathbf{\hat{P}}_{(-\infty,b_{j}(t)]}-\mathbf{\hat{J}}_{t}(a_{j}(t))\mathbf{\hat{P}}_{(-\infty,a_{j}(t)]}- \int_{\Delta_{j,t}}\mathbf{\hat{J}}_{t}^{\prime}(y) \mathbf{\hat{P}}_{(-\infty,y]} dy\bigg\|_{1}\leq
\end{equation}
\begin{equation}
\big\|\mathbf{\hat{J}}_{t}(b_{j}(t))\mathbf{\hat{P}}_{(-\infty,b_{j}(t)]}\big\|_{1}+\big\|\mathbf{\hat{J}}_{t}(a_{j}(t))\mathbf{\hat{P}}_{(-\infty,a_{j}(t)]}\big\|_{1}+\bigg\| \int_{\Delta_{j,t}}\mathbf{\hat{J}}_{t}^{\prime}(y) \mathbf{\hat{P}}_{(-\infty,y]} dy\bigg\|_{1}\leq
\end{equation}
\begin{equation}
\big\|\mathbf{\hat{J}}_{t}(b_{j}(t))\big\|_{1}\|\mathbf{\hat{P}}_{(-\infty,b_{j}(t)]}\big\|+\big\|\mathbf{\hat{J}}_{t}(a_{j}(t))\big\|_{1}\big\|\mathbf{\hat{P}}_{(-\infty,a_{j}(t)]}
\big\|+\int_{\Delta_{j,t}}\big\|\mathbf{\hat{J}}_{t}(y)^\prime \mathbf{\hat{P}}_{(-\infty,y]}\big\|_{1} dy\leq
\end{equation}
\begin{equation}
\big\|\mathbf{\hat{J}}_{t}(b_{j}(t))\big\|_{1}+\big\|\mathbf{\hat{J}}_{t}(a_{j}(t))\big\|_{1}+\int_{\Delta_{j,t}}\big\|\mathbf{\hat{J}}_{t}^\prime(y)\big\|_{1}\big\| \mathbf{\hat{P}}_{(-\infty,y]}\big\|dy =
\end{equation}
\begin{equation}
\big\|\mathbf{\hat{J}}_{t}(b_{j}(t))\big\|_{1}+\big\|\mathbf{\hat{J}}_{t}(a_{j}(t))\big\|_{1}+\int_{\Delta_{j,t}}\big\|\mathbf{\hat{J}}_{t}^\prime(y)\big\|_{1}dy \leq 
\end{equation}
\begin{equation}
\big\|\mathbf{\hat{J}}_{t}(b_{j}(t))\big\|_{1}+\big\|\mathbf{\hat{J}}_{t}(a_{j}(t))\big\|_{1}+|\Delta_{j,t}|\sup_{y\in\Delta_{j,t}}\big\|\mathbf{\hat{J}}_{t}^\prime(y)\big\|_{1}dy \leq 
\end{equation}
\begin{equation}
\sup_{x\in \Delta_{i,t}}|\Gamma(t,x,b_{j}(t))|+\sup_{x\in \Delta_{i,t}}|\Gamma(t,x,a_{j}(t))|+ |\Delta_{j,t}|\sup_{\overset{x\in \Delta_{i,t}}{ y\in \Delta_{j,t}}}|\partial_{y}\Gamma(t,x,y)|\leq
\end{equation}
\begin{equation}
\sup_{\overset{(x,y)\in}{ \Delta_{i,t}\times \Delta_{j,t}}}\bigg(2|\Gamma(t,x,y)|+|\Delta_{j,t}\Gamma^{'}(t,x,y)|\bigg)
\end{equation}
\end{proof}
\subsection{Another Way to Estimate the Off-diagonal Terms (\ref{eqn:uzbek9})}
\;\;\; In the previous section, the kernel $\Gamma(t,x,y)$ characterizing the non-unitary evolution of the density operator in the hypothesis of Theorem \ref{eqn:theoremkupsch} was treated in rather general terms. However, if more is known about the kernel $\Gamma(t,x,y)$, then one may employ another technique to estimate $\big\|\mathbf{\hat{P}}_{\Delta_{i,t}}\mathscr{E}_{t}\big(\boldsymbol{\hat{\rho}}_{S_{0}}\big)\mathbf{\hat{P}}_{\Delta_{j,t}}\big\|_{1} $. Namely, we will utilize the following theorem from \cite{stolz} in order to bound the non-diagonal terms for the case where it is known how to express $\mathbf{\hat{P}}_{\Delta_{i,t}}\mathscr{E}_{t}\big(\boldsymbol{\hat{\rho}}_{S_{0}}\big)\mathbf{\hat{P}}_{\Delta_{j,t}}$ as a product of two trace class operators acting in the Hilbert space $L^{2}\big(M,\mu\big)$. We will focus on the case where $M= \mathbb{R}$ and $\mu$ is the Lebesgue measure.

\begin{definition4}
\label{eqn:stolz}

 Let $A(x,z)$ and $B(z,x)$ $\in L^{2}(\mathbb{R})\:\forall \:z\in \mathbb{R}$ and
 \begin{equation}
\int\|A(\cdot,z)\|_{L^{2}(\mathbb{R})}\|B(z,\cdot)\|_{L^{2}(\mathbb{R})}dz <\infty .
\end{equation}
Then there is a trace class operator $\mathbf{\hat{A}}\mathbf{\hat{B}}$ acting in $L^{2}(\mathbb{R})$ with kernel 
\begin{equation}
AB(x,y)= \int A(x,z)B(z,y)dz
\end{equation}
such that 
\begin{equation}
 \big\|\mathbf{\hat{A}}\mathbf{\hat{B}}\big\|_{1}\leq \int\|A(\cdot,z)\|_{L^{2}(\mathbb{R})}\|B(z,\cdot)\|_{L^{2}(\mathbb{R})}dz
 \end{equation}
The operators $\mathbf{\hat{A}}$ and $\mathbf{\hat{B}}$ also act in $L^{2}(\mathbb{R})$ and have respective kernels $A(x,y)$ and $B(x,y)$.
\end{definition4}
The first step towards applying Theorem \ref{eqn:stolz} to the estimation of the terms in (\ref{eqn:kupschbound}) is rewriting $\mathbf{\hat{P}}_{\Delta_{i,t}}\mathscr{E}_{t}\big(\boldsymbol{\hat{\rho}}_{S_{0}}\big)\mathbf{\hat{P}}_{\Delta_{j,t}}$ as a product of of two operators. Let us fix $t$, and note that the main challenge is the kernel $\Gamma(t,x,y)$. Expanding, we see that
\begin{equation}
\label{eqn:montreal1}
\mathbf{\hat{P}}_{\Delta_{i,t}}\mathscr{E}_{t}\big(\boldsymbol{\hat{\rho}}_{S_{0}}\big)\mathbf{\hat{P}}_{\Delta_{j,t}} = \int\int \Gamma(t,x,y)\sum_{n}p_{n}\psi_{n,0}(x)\psi_{n,0}^{*}(y)\chi_{\Delta_{i,t}}(x)\chi_{\Delta_{j,t}}(y)\big|x\big\rangle\big\langle y\big| dxdy.
\end{equation}
where $K(x,y) = \sum_{n}p_{n}\psi^{*}_{n,0}(x)\psi_{n,0}(y)$ is the kernel of $\boldsymbol{\hat{\rho}}_{S_{0}}= \sum_{n}p_{n}\big|\psi_{n,0}\big\rangle\big\langle\psi_{n,0}\big|$. 
If the kernel $\Gamma(t,x,y)$ were to have a decomposition of the form
\begin{equation}
\Gamma(t,x,y) = \int \phi(t,x,z)\eta(t,y,z) dz
\end{equation}
with 
\begin{equation}
\label{eqn:montreal2}
\int|\eta(t,x,z)\phi(t,y,z)|^{2}dz<\infty
\end{equation}
for all $ (t,x,y) \in\mathbb{R}^{2}$, then Theorem \ref{eqn:stolz} would be applicable for any $t$. Assuming (\ref{eqn:montreal2}) we have
\begin{equation}
\big\|\mathbf{\hat{P}}_{\Delta_{i,t}}\mathscr{E}_{t}\big(\boldsymbol{\hat{\rho}}_{S_{0}}\big)\mathbf{\hat{P}}_{\Delta_{j,t}}\big\|_{1} \leq 
\sum_{n}p_{n}\big\|\mathbf{\hat{P}}_{\Delta_{i,t}}\mathscr{E}_{t}\big(\big|\psi_{n,0}\big\rangle \big\langle \psi_{n,0}\big|\big)\mathbf{\hat{P}}_{\Delta_{j,t}}\big\|_{1} \leq 
\end{equation}
\begin{equation}
=\sum_{n}p_{n}\int\Bigg(\bigg[\int |\phi(t,x,z)\psi_{n,0}(x)\chi_{\Delta_{i,t}}(x)|^{2}dx\bigg]\bigg[\int |\eta(t,y,z)\psi_{n,0}(y)\chi_{\Delta_{j,t}}(y)|^{2}dy\bigg]\Bigg)dz = 
\end{equation}
\begin{equation}
\sum_{n}p_{n}\int \int\Bigg(\bigg[\int |\phi(t,x,z)\eta(t,y,z)|^{2}dz\bigg]|\psi_{n,0}(x)\chi_{\Delta_{i,t}}(x)|^{2}|\psi_{n,0}(y)\chi_{\Delta_{j,t}}(y)|^{2}\Bigg)dxdy\leq
\end{equation}
\begin{equation}
\sum_{n}p_{n}\bigg[\max_{(x,y)\in \mathbb{R}^{2}} \int|\phi(t,x,z)\eta(t,y,z)|^{2}dz\bigg]\int \int\Bigg(|\psi_{n,0}(x)\chi_{\Delta_{i,t}}(x)|^{2}|\psi_{n,0}(y)\chi_{\Delta_{j,t}}(y)|^{2}\Bigg)dxdy=  
\end{equation}
\begin{equation}
\sum_{n}p_{n}\bigg[\max_{(x,y) \mathbb{R}^{2}}  \int|\phi(t,x,z)\eta(t,y,z)|^{2}dz\bigg]\Big(\int_{\Delta_{i,t}}|\psi_{n,0}(x)|^{2}dx\Big)\Big(\int_{\Delta_{j,t}}|\psi_{n,0}(y)|^{2}dy\Big)\leq
\end{equation}
\begin{equation}
\max_{(x,y)\in \mathbb{R}^{2}} \int|\phi(t,x,z)\eta(t,y,z)|^{2}dz<\infty  
\end{equation}
The Hypothesis of Theorem \ref{eqn:stolz} is therefore satisfied, for all $t>0$, and we conclude that 
\begin{equation}
\label{eqn:montreal3}
 \big\| \mathbf{\hat{P}}_{\Delta_{i,t}}\mathscr{E}_{t}\big(\boldsymbol{\hat{\rho}}_{S_{0}}\big)\mathbf{\hat{P}}_{\Delta_{j,t}}  \big\|_{1} \leq   \sum_{n}p_{n}\int\|A^{n}_{i,t}(\cdot,z)\|_{L^{2}(\mathbb{R})}\|B^{n}_{j,t}(z,\cdot)\|_{L^{2}(\mathbb{R})}dz
\end{equation}
with $A^{n}_{i,t}(x,z):= \phi(t,x,z)\psi_{n,0}(x)\chi_{\Delta_{i,t}}(x)$ and $B^{n}_{j,t}(z,x):= \eta(t,x,z)\psi^{*}_{n,0}(x)\chi_{\Delta_{j,t}}(x)$.\\

We now formalize the above as a corollary to Theorem \ref{eqn:stolz}.
\begin{Co}
\label{eqn:montrealuzbek}
If the kernel $\Gamma(t,x,y)$ has a decomposition 
\begin{equation}
\Gamma(t,x,y) = \int \phi(t,x,z)\eta(t,y,z) dz
\end{equation}
with 
\begin{equation}
\int|\eta(t,x,z)\phi(t,y,z)|^{2}dz<\infty
\end{equation}
for all $ (t,x,y) \in [0,\infty)\times\mathbb{R}^{2}$, then 

\begin{equation}
 \big\| \mathbf{\hat{P}}_{\Delta_{i,t}}\mathscr{E}_{t}\big(\boldsymbol{\hat{\rho}}_{S_{0}}\big)\mathbf{\hat{P}}_{\Delta_{j,t}}  \big\|_{1} \leq   \sum_{n}p_{n}\int\|A^{n}_{i,t}(\cdot,z)\|_{L^{2}(\mathbb{R})}\|B^{n}_{j,t}(z,\cdot)\|_{L^{2}(\mathbb{R})}dz
\end{equation}
with $A^{n}_{i,t}(x,z):= \phi(t,x,z)\psi_{n,0}(x)\chi_{\Delta_{i,t}}(x)$ and $B^{n}_{j,t}(z,x):= \eta(t,x,z)\psi^{*}_{n,0}(x)\chi_{\Delta_{j,t}}(x)$.
\end{Co}
\begin{proof}
The proof is in the preceding discussion concluding with equation (\ref{eqn:montreal3}).
\end{proof}
\subsubsection{Example}
For the case where the kernel $\Gamma(t,x,y) = e^{-t^{n}\alpha(x-y)^{2}}$, $n> 0,t>0$, we may express such a function as a convolution of Gaussians. Namely, 
\begin{equation}
\Gamma(t,x,y) = e^{-t^{n}\alpha(x-y)^{2}} = 2\sqrt{\frac{t^{n}\alpha}{\pi}}\int e^{-2t^{n}\alpha (x-z)^{2}}e^{-2t^{n}\alpha (y-z)^{2}}dz. 
\end{equation}
In this case the $\phi$ and $\eta$ from (\ref{eqn:montreal2}) are just 
\begin{equation}
\phi(t,x,z)= \sqrt{2\sqrt{\frac{t^{n}\alpha}{\pi}}}e^{-2t^{n}\alpha (x-z)^{2}}
\end{equation}
and 
\begin{equation}
 \eta(t,y,z)= \sqrt{2\sqrt{\frac{t^{n}\alpha}{\pi}}}e^{-2t^{n}\alpha (y-z)^{2}}. 
\end{equation}
Notice that here 
\begin{equation}
\int|\phi(t,x,z)\eta(t,y,z)|^{2}dz = \frac{4t^{n}}{\pi}\int e^{-4t^{n}\alpha (x-z)^{2}}e^{-4t^{n}\alpha (y-z)^{2}} dz = \frac{4t^{n}\alpha}{\pi}\frac{\sqrt{\pi}}{2^{\frac{3}{2}}\sqrt{t^{n}\alpha}}e^{-2t^{n}\alpha(x-y)^{2}}=
\end{equation}
\begin{equation}
\frac{\sqrt{2t^{n}\alpha}}{\sqrt{\pi}}e^{-2t^{n}\alpha(x-y)^{2}} <\infty
\end{equation}
for all $t>0$ and $(x,y)\in \mathbb{R}^{2}$. With these $\phi$ and $\eta$, we can bound the left-hand side of (\ref{eqn:montreal3}) as follows. 
\begin{equation}
\big\| \mathbf{\hat{P}}_{\Delta_{i,t}}\mathscr{E}_{t}\big(\boldsymbol{\hat{\rho}}_{S_{0}}\big)\mathbf{\hat{P}}_{\Delta_{j,t}}  \big\|_{1} \leq 
\end{equation}
\begin{equation}
\sum_{n}p_{n}\int \int\Bigg(\bigg[\int |\phi(t,x,z)\eta(t,y,z)|^{2}dz\bigg]|\psi_{n,0}(x)\chi_{\Delta_{i,t}}(x)|^{2}|\psi_{n,0}(y)\chi_{\Delta_{j,t}}(y)|^{2}\Bigg)dxdy = 
\end{equation}
\begin{equation}
\sum_{n}p_{n} \int \int\Bigg(\frac{\sqrt{2t^{n}\alpha}}{\sqrt{\pi}}e^{-2t^{n}\alpha(x-y)^{2}}|\psi_{n,0}(x)\chi_{\Delta_{i,t}}(x)|^{2}|\psi_{n,0}(y)\chi_{\Delta_{j,t}}(y)|^{2}\Bigg)dxdy  = 
\end{equation}
\begin{equation}
\label{eqn:montreal4}
\sum_{n}p_{n}\int_{\Delta_{i,t}} \int_{\Delta_{j,t}}\Bigg(\frac{\sqrt{2t^{n}\alpha}}{\sqrt{\pi}}e^{-2t^{n}\alpha(x-y)^{2}}|\psi_{n,0}(x)|^{2}|\psi_{n,0}(y)|^{2}\Bigg)dxdy
\end{equation}
The kernels $\frac{\sqrt{2t^{n}\alpha}}{\sqrt{\pi}}e^{-2t^{n}\alpha(x-y)^{2}}$ form a sequence converging to the delta function as $t \to \infty$. Hence, whenever $x\in \Delta_{i,t}$, $y\in\Delta_{j,t}$, and (since $\Delta_{i,t}\cap\Delta_{j,t}=\emptyset$) (\ref{eqn:montreal4}) vanishes as $t\rightarrow \infty$.\\

The importance of Gaussian states warrants summarizing the above discussion as another corollary of Theorem \ref{eqn:stolz}.
\begin{Co}[\textbf{Theorem} \ref{eqn:stolz} \textbf{with Gaussian assumptions for} $\Gamma(t,x,y)$ ]
\label{eqn:montrealuzbek2}
Fix $t>0$. Now let 
\begin{equation}
\Gamma(t,x,y) = e^{-t^{n}\alpha(x-y)^{2}} = 2\sqrt{\frac{t^{n}\alpha}{\pi}}\int e^{-2t^{n}\alpha (x-z)^{2}}e^{-2t^{n}\alpha (y-z)^{2}}dz
\end{equation}
where $n>0$, and assume that $\Delta_{i,t}\cap\Delta_{j,t}=\emptyset$, 
then
\begin{equation}
1)\;\;\big\| \mathbf{\hat{P}}_{\Delta_{i,t}}\mathscr{E}_{t}\big(\boldsymbol{\hat{\rho}}_{S_{0}}\big)\mathbf{\hat{P}}_{\Delta_{j,t}}  \big\|_{1} \leq
\sum_{n}p_{n}\int_{\Delta_{i,t}} \int_{\Delta_{j,t}}\Bigg(\frac{\sqrt{2t^{n}\alpha}}{\sqrt{\pi}}e^{-2t^{n}\alpha(x-y)^{2}}|\psi_{n,0}(x)|^{2}|\psi_{n,0}(y)|^{2}\Bigg)dxdy
\end{equation}
\begin{equation}
2)\;\;\lim_{t\rightarrow\infty}\big\|\mathbf{\hat{P}}_{\Delta_{i,t}}\mathscr{E}_{t}\big(\boldsymbol{\hat{\rho}}_{S_{0}}\big)\mathbf{\hat{P}}_{\Delta_{j,t}}\big\|_{1} = 0
\end{equation}
\end{Co}
\begin{proof}
The proof is in the preceding discussion concluding with equation (\ref{eqn:montreal4}).
\end{proof}

\section{Estimating the Diagonal Terms (\ref{eqn:uzbek8})}
\label{eqn:section4.2}
\;\;\; We have hitherto developed the tools necessary to estimate the trace norm of the off-diagonal terms (\ref{eqn:uzbek9}) arising from the optimization problem (\ref{eqn:uzbek5}). We now need to study the diagonal terms (\ref{eqn:uzbek8}), i.e. estimate
\begin{equation}
\label{eqn:thediag}
\min_{PVM}\bigg\|\sum_{i}\Big(\mathbf{\hat{P}}_{\Delta_{i,t}}\otimes \mathbb{I}\Big) \boldsymbol{\hat{\rho}}_{t}\Big(\mathbf{\hat{P}}_{\Delta_{i,t}}\otimes \mathbb{I}\Big)- \frac{1}{\mathscr{N}(t)}\sum_{i}\Big(\mathbf{\hat{P}}_{\Delta_{i,t}}\otimes \mathbf{\hat{P}}^{E,t}_{i}\Big) \boldsymbol{\hat{\rho}}_{t}\Big(\mathbf{\hat{P}}_{\Delta_{i,t}}\otimes \mathbf{\hat{P}}^{E,t}_{i}\Big)    \bigg\|_{1}
\end{equation}
The minimization is taken over all PVMs resolving the identity operator in the space associated with the environmental degrees of freedom.
Recall that $\sum_{i}\sum_{j}\Big(\mathbf{\hat{P}}_{\Delta_{i,t}}\otimes \mathbb{I}\Big)\boldsymbol{\hat{\rho}}_{t}\Big(\mathbf{\hat{P}}_{\Delta_{j,t}}\otimes \mathbb{I}\Big)$ is just another way of writing $\boldsymbol{\hat{\rho}}_{t}$ (see (\ref{eqn:partitionuzbek})). The use of the PVM $\big\{\mathbf{\hat{P}}_{\Delta_{i,t}}\big\}_{i}$ in the first term in (\ref{eqn:thediag}) is just technical. However, the usage of the same PVM in the second term implies a measurement of the von Neumann type performed on the system. When estimating the off-diagonal terms (\ref{eqn:uzbek9}), we only needed to study their asymptotic behavior in $t$, since the families of PVM acting on the system's degree of freedom, $\mathbf{\hat{P}}_{\Delta_{i,t}}$ were assumed to be predetermined. In the case of the diagonal terms (\ref{eqn:uzbek8}), the PVM will vary with $t$, which poses a more challenging $t$-dependent optimization problem.\\

The reader might have already noted that the map $\sum_{i}\mathbf{\hat{P}}_{\Delta_{i,t}}\otimes \mathbf{\hat{P}}^{E,t}_{i}\big(...\big)\mathbf{\hat{P}}_{\Delta_{i,t}}\otimes \mathbf{\hat{P}}^{E,t}_{i}$ is unlike measurements of von Neumann type seen in 2.2.5 of \cite{Nielsen} since it does not preserve the trace. Both are indeed completely positive maps, but the latter map turns out to reduce the trace in general, i.e. 
\begin{equation}
Tr\bigg\{\sum_{i}\Big(\mathbf{\hat{P}}_{\Delta_{i,t}}\otimes \mathbf{\hat{P}}^{E,t}_{i}\Big)\boldsymbol{\hat{\rho}}_{t}\Big(\mathbf{\hat{P}}_{\Delta_{i,t}}\otimes \mathbf{\hat{P}}^{E,t}_{i}\Big)\bigg\}\leq Tr\{\boldsymbol{\hat{\rho}}_{t}\} = 1. \end{equation}
Indeed, the PVM $\{\mathbf{\hat{P}}_{\Delta_{i}}\otimes \mathbf{\hat{P}}^{E,t}_{i}\}_{i}$ by itself does not describe a measurement for the product of the system's and environment's Hilbert spaces because it does not resolve the identity operator acting over the entire Hilbert space $\mathscr{H}_{S}\otimes\mathscr{H}_{E}$ but rather the identity of a subspace of $\mathscr{H}_{S}\otimes\mathscr{H}_{E}$; hence the need for the normalization constant $\mathscr{N}(t)$ in (\ref{eqn:thediag}). The associated PVM, which preserves trace, and resolves the identity of $\mathscr{H}_{S}\otimes\mathscr{H}_{E}$, is the family of projectors $\big\{\mathbf{\hat{P}}_{\Delta_{i,t}}\otimes \mathbf{\hat{P}}^{E,t}_{j}\big\}_{i,j}$. This set includes outcomes pertaining to the case where the environment $E$ is measured to be in the state labeled by the index $j$ which differs from the outcome $i\neq j$ measured by the system $S$. Hence, we exclude the $i\neq j$ terms when constructing the approximating SBS state (\ref{eqn:SBSCVuzbek}). 

Let us now estimate (\ref{eqn:thediag}). We begin by rewriting the operator $\sum_{i}\Big(\mathbf{\hat{P}}_{\Delta_{i,t}}\otimes \mathbb{I}\Big)\boldsymbol{\hat{\rho}}_{t}\Big(\mathbf{\hat{P}}_{\Delta_{i,t}}\otimes \mathbb{I}\Big)$ in the form described by the following Lemma:
\begin{definition2}[\textbf{Rewriting} (\ref{eqn:thediag})]
\label{eqn:rewrittingthingstask}
\begin{equation}
\sum_{i}\Big(\mathbf{\hat{P}}_{\Delta_{i,t}}\otimes \mathbb{I}\Big)\boldsymbol{\hat{\rho}}_{t}\Big(\mathbf{\hat{P}}_{\Delta_{i,t}}\otimes \mathbb{I}\Big) = \sum_{i}\bar{p}_{i}(t)\mathscr{U}_{t}\big(\mathscr{E}_{t}\big(\boldsymbol{\hat{\rho}}_{S_{i,t}}\big)\otimes \boldsymbol{\hat{\rho}}^{E,0}\big)
\end{equation}
where $\boldsymbol{\hat{\rho}}_{t}$ is given by \ref{eqn:jarekcom} and 
\begin{equation}
\boldsymbol{\hat{\rho}}_{S_{i,t}}=\int_{\mathbb{R}}\int_{\mathbb{R}} \sum_{n}q_{n}K^{n}_{i,t}(x,y)|x\rangle\langle y|dxdy
\end{equation}
\begin{equation}
K^{n}_{i,t}(x,y):= \frac{\mathbbm{1}_{\Delta_{i,t}}(x)\psi_{n,0}(x)}{\sqrt{\bar{p}_{i,n}(t)}}\frac{\mathbbm{1}_{\Delta_{i,t}}(y)\psi^{*}_{n,0}(y)}{\sqrt{\bar{p}_{i,n}(t)}}
\end{equation}
\begin{equation}
\bar{p}_{i,n}(t) := \int_{\Delta_{i,t}}K^{n}(x,x)dx,
\end{equation}
\begin{equation}
K^{n}(x,y)  := \psi_{n,0}(x)\psi^{*}_{n,0}(y)\end{equation}
Recalling that 
\begin{equation}
\mathscr{U}_{t}(\mathbf{\hat{A}}):= e^{-it\gamma \mathbf{\hat{X}}\otimes \mathbf{\hat{B}}}\mathbf{\hat{A}}e^{it\gamma \mathbf{\hat{X}}\otimes \mathbf{\hat{B}}}
\end{equation}
\end{definition2}

\begin{proof}
At $t=0$ $\boldsymbol{\hat{\rho}}_{S_{0}}$ is in general a mixture. Let
\begin{equation}
\boldsymbol{\hat{\rho}}_{S_{0}} = \int\int\Big(\sum_{n}p_{n}K^{n}(x,y)\Big)\big|x\big\rangle\big\langle y\big|dxdy
\end{equation}
Hence, 
\begin{equation}
\label{eqn:thestart}
\sum_{i}\Big(\mathbf{\hat{P}}_{\Delta_{i,t}}\otimes \mathbb{I}\Big)\boldsymbol{\hat{\rho}}_{t}\Big(\mathbf{\hat{P}}_{\Delta_{i,t}}\otimes \mathbb{I}\Big)= \sum_{i}\int_{\Delta_{i,t}}\int_{\Delta_{i,t}}\Big(\sum_{n}p_{n}K^{n}(x,y)\Big)\Gamma(t,x,y)\big|x\big\rangle\big\langle y\big|\otimes \boldsymbol{\hat{\rho}}^{E,t}_{x,y}dxdy  = 
\end{equation}
\begin{equation}
\label{eqn:thediagrecast}
\sum_{n}q_{n}\sum_{i}\bar{p}_{i,n}(t)\int_{\Delta_{i,t}}\int_{\Delta_{i,t}}\frac{K^{n}(x,y)}{\bar{p}_{i,n}(t)}\Gamma(t,x,y)|x\rangle\langle y|\otimes\boldsymbol{\hat{\rho}}^{E,t}_{x,y}dxdy= 
\end{equation}
where again
\begin{equation}
\bar{p}_{i,n}(t) := \int_{\Delta_{i,t}}K^{n}(x,x)dx.
\end{equation}
That is, 
\begin{equation}
\label{eqn:thediagrecast2}
(\ref{eqn:thediagrecast}) = \sum_{n}q_{n}\sum_{i}\bar{p}_{i,n}(t)\int_{\mathbb{R}}\int_{\mathbb{R}} K^{n}_{i,t}(x,y)\Gamma(t,x,y)|x\rangle\langle y|\otimes \boldsymbol{\hat{\rho}}^{E,t}_{x,y}dxdy
\end{equation}
Recalling that
$K^{n}(x,y)=\psi_{n,0}(x)\psi^{*}_{n,0}(y)$ we now define 
\begin{equation}
K^{n}_{i,t}(x,y):= \mathbbm{1}_{\Delta_{i,t}}(x)\mathbbm{1}_{\Delta_{i,t}}(y)\frac{K^{n}(x,y)}{\bar{p}_{i,n}(t)} = \frac{\mathbbm{1}_{\Delta_{i,t}}(x)\psi_{n,0}(x)}{\sqrt{\bar{p}_{i,n}(t)}}\frac{\mathbbm{1}_{\Delta_{i,t}}(y)\psi_{n,0}^{*}(y)}{\sqrt{\bar{p}_{i,n}(t)}}.
\end{equation} 
Furthermore, define
\begin{equation}
\label{eqn:joy4}
\psi_{S^{n}_{i,t}}(x):=\frac{\mathbbm{1}_{\Delta_{i,t}}(x)\psi_{n,0}(x)}{\sqrt{\bar{p}_{i,n}(t)}} 
\end{equation}
and write 
\begin{equation}
\label{eqn:braketshinanigans}
K^{n}_{i,t}(x,y) =\psi_{S^{n}_{i,t}}(x)\psi_{S^{n}_{i,t}}^{*}(y)\;\; \big(Kernel\;\;defining\;\;\boldsymbol{\hat{\rho}}_{S_{i,t}^{n}}\;\;)
\end{equation}
Finally, 
\begin{equation}
(\ref{eqn:thediagrecast2}) = \sum_{n}q_{n}\sum_{i}\bar{p}_{i}(t)\mathscr{U}_{t}\big(\mathscr{E}_{t}\big(\boldsymbol{\hat{\rho}}_{S^{n}_{i,t}}\big)\otimes \boldsymbol{\hat{\rho}}^{E,0}\big)
\end{equation}
\begin{equation}
\sum_{i}\bar{p}_{i}(t)\mathscr{U}_{t}\big(\mathscr{E}_{t}\big(\sum_{n}q_{n}\boldsymbol{\hat{\rho}}_{S^{n}_{i,t}}\big)\otimes \boldsymbol{\hat{\rho}}^{E,0}\big) =
\end{equation}
\begin{equation}
\sum_{i}\bar{p}_{i}(t)\mathscr{U}_{t}\big(\mathscr{E}_{t}\big(\boldsymbol{\hat{\rho}}_{S_{i,t}}\big)\otimes \boldsymbol{\hat{\rho}}^{E,0}\big) =
\end{equation}
where we define
\begin{equation}
\label{eqn:joy}
\boldsymbol{\hat{\rho}}_{S_{i,t}}:= \int_{\mathbb{R}}\int_{\mathbb{R}} \sum_{n}q_{n}K^{n}_{i,t}(x,y)\big|x\big\rangle\big\langle y\big|dxdy = \frac{\mathbf{\hat{P}}_{\Delta_{i,t}}\boldsymbol{\hat{\rho}}_{S_{0}}\mathbf{\hat{P}}_{\Delta_{i,t}}}{Tr\big\{ \mathbf{\hat{P}}_{\Delta_{i,t}}\boldsymbol{\hat{\rho}}_{S_{0}}\mathbf{\hat{P}}_{\Delta_{i,t}}\big\}} 
\end{equation}
and apply the definition of $\mathscr{U}_{t}$ presented in (\ref{eqn:coca4}) with $g = \gamma$ and $N_{E}=1$, i.e. 
\begin{equation}
\mathscr{U}_{t}(\mathbf{\hat{A}}):= e^{-it\gamma \mathbf{\hat{X}}\otimes \mathbf{\hat{B}}}\mathbf{\hat{A}}e^{it\gamma \mathbf{\hat{X}}\otimes \mathbf{\hat{B}}}.
\end{equation}
\end{proof}
Using Lemma \ref{eqn:rewrittingthingstask} we may now write 
\begin{equation}
\label{eqn:thestart2}
\sum_{i}\Big(\mathbf{\hat{P}}_{\Delta_{i,t}}\otimes \mathbf{\hat{P}}^{E,t}_{i}\Big)\boldsymbol{\hat{\rho}}_{t}\Big(\mathbf{\hat{P}}_{\Delta_{i,t}}\otimes \mathbf{\hat{P}}^{E,t}_{i}\Big) = \sum_{i}\bar{p}_{i}(t)\Big(\mathbb{I}\otimes \mathbf{\hat{P}}^{E,t}_{i}\Big)\mathscr{U}_{t}\big(\mathscr{E}_{t}\big(\boldsymbol{\hat{\rho}}_{S_{i,t}}\big)\otimes\boldsymbol{\hat{\rho}}^{E,0}\big)\Big(\mathbb{I}\otimes \mathbf{\hat{P}}^{E,t}_{i}\Big)
\end{equation}
Finally, normalizing the operator (\ref{eqn:thestart2}) we obtain an SBSCV state which approximates $\boldsymbol{\hat{\rho}}_{t}$.
\begin{equation}
\label{eqn:thestart3}
\boldsymbol{\hat{\rho}}_{SBSCV,t}:=\frac{1}{\mathscr{N}(t)}\sum_{i}\bar{p}_{i}(t)\Big(\mathbb{I}\otimes \mathbf{\hat{P}}^{E,t}_{i}\Big)\mathscr{U}_{t}\big(\mathscr{E}_{t}\big(\boldsymbol{\hat{\rho}}_{S_{i,t}}\big)\otimes\boldsymbol{\hat{\rho}}^{E,0}\big)\Big(\mathbb{I}\otimes \mathbf{\hat{P}}^{E,t}_{i}\Big).
\end{equation}
To proceed, we will need the following lemma. 
\begin{definition2}[\textbf{Trace Distance Lemma}]
\label{eqn:reversetin}
\vspace{3mm}
$\big\|\boldsymbol{\hat{\rho}}-\eta\boldsymbol{\hat{\sigma}}\big\|_{1}\leq L$ implies $\big\|\boldsymbol{\hat{\rho}}-\boldsymbol{\hat{\sigma}}\big\|_{1}\leq 2L$ for constants $L\geq 0$ and $\eta \in [0,1]$.
\end{definition2}
\begin{proof}
Using reverse triangle inequality we see that 
\begin{equation}
L\geq \|\boldsymbol{\hat{\rho}}-\eta\boldsymbol{\hat{\sigma}}\|_{1}\geq \big|\|\boldsymbol{\hat{\rho}}\|_{1}-\|\eta\boldsymbol{\hat{\sigma}}\|_{1}\big| = \|\boldsymbol{\hat{\rho}}\|_{1}-\|\eta\boldsymbol{\hat{\sigma}}\|_{1} = 1-\eta
\end{equation}
furthermore
\begin{equation}
 \|\boldsymbol{\hat{\rho}}-\boldsymbol{\hat{\sigma}}\|_{1} =   \|\boldsymbol{\hat{\rho}}-\eta\boldsymbol{\hat{\sigma}}+\eta\boldsymbol{\hat{\sigma}}-\hat{\sigma}\|_{1}   \leq   \|\boldsymbol{\hat{\rho}}-\eta\boldsymbol{\hat{\sigma}}\|_{1}+\|\eta\boldsymbol{\hat{\sigma}}-\boldsymbol{\hat{\sigma}}\|_{1} \leq
\end{equation}
\begin{equation}
L +(1-\eta)\|\hat{\sigma}\|_{1} = L+(1-\eta)\leq L+L= 2L
\end{equation}
\end{proof}
We will be estimating $\big\|\boldsymbol{\hat{\rho}}_{t}-\mathscr{N}(t)\boldsymbol{\hat{\rho}}_{SBSCV,t}\big\|_{1}$ first, and then, using Lemma \ref{eqn:reversetin} we shall bound $\big\|\boldsymbol{\hat{\rho}}_{t}-\boldsymbol{\hat{\rho}}_{SBSCV,t}\big\|_{1}$.

The representation (\ref{eqn:thestart3}) makes transparent the structure of the dynamics being imposed on the total initial states $\boldsymbol{\hat{\rho}}_{S_{0}}\otimes\boldsymbol{\hat{\rho}}^{E,0}$ by making explicit all of the quantum maps generating the dynamics, i.e. 
\begin{equation}
\boldsymbol{\hat{\rho}}_{SBSCV,t}=\Lambda_{t}\circ\mathscr{U}_{t}\circ\big(\mathscr{E}_{t}\otimes\mathcal{I}_{E}\big)\Big(\boldsymbol{\hat{\rho}}_{S_{0}}\otimes\boldsymbol{\hat{\rho}}^{E,0}\Big)
\end{equation}
where $\mathcal{I}_{E}$ is the identity map on the environmental degrees of freedom and 
\begin{equation}
\Lambda_{t}(\mathbf{\hat{A}}):=\sum_{i}\frac{1}{\mathscr{N}(t)}\sum_{i}\bar{p}_{i}(t)\Big(\mathbb{I}\otimes \mathbf{\hat{P}}^{E,t}_{i}\Big)\mathbf{\hat{A}}\Big(\mathbb{I}\otimes \mathbf{\hat{P}}^{E,t}_{i}\Big).
\end{equation}
It is clear that the quantum maps $\Lambda_{t}$ and $\mathscr{E}_{t}\otimes\mathcal{I}_{E}$ commute since they respectively act non-trivially on different factors of a tensor product space. What is more interesting, and less obvious, is the commutativity between $\mathscr{E}_{t}\otimes\mathcal{I}_{E}$ and the unitary map $\mathscr{U}_{t}$. Proving this is the content of the following lemma (Lemma \ref{eqn:removingmaps}); a lemma that we shall need for the proof of the main result of this work. 

\begin{definition2}[\textbf{Commutation of} $\mathscr{E}_{t}\otimes\mathcal{I}_{E}$ \textbf{and} $\mathscr{U}_{t}$]
\label{eqn:removingmaps}
\begin{equation}
\mathscr{U}_{t}\circ\big(\mathscr{E}_{t}\otimes\mathcal{I}_{E}\big)\Big(\boldsymbol{\hat{\rho}}_{S_{0}}\otimes\boldsymbol{\hat{\rho}}^{E,0}\Big) = \big(\mathscr{E}_{t}\otimes\mathcal{I}_{E}\big)\circ\mathscr{U}_{t}\Big(\boldsymbol{\hat{\rho}}_{S_{0}}\otimes\boldsymbol{\hat{\rho}}^{E,0}\Big)
\end{equation}
\end{definition2}
\begin{proof}
We remind the reader that we are always working within the framework discussed in the introduction to this paper (see equations (\ref{eqn:intham2}) through (\ref{eqn:evo12})). With this in mind, define $\mathscr{V}_{t}\big(\mathbf{\hat{A}}\big):= e^{-it\gamma^{\prime}\mathbf{\hat{X}}\otimes\mathbb{I}_{E}\otimes\mathbf{\hat{B}}^{\prime}}\big( \mathbf{\hat{A}}\big)e^{it\gamma^{\prime}\mathbf{\hat{X}}\otimes\mathbb{I}_{E}\otimes\mathbf{\hat{B}}^{\prime}}$. Then,
\begin{equation}
\mathscr{U}_{t}\circ\big(\mathscr{E}_{t}\otimes\mathcal{I}_{E}\big)\Big(\boldsymbol{\hat{\rho}}_{S_{0}}\otimes\boldsymbol{\hat{\rho}}^{E,0}\Big) = \mathscr{U}_{t}\Bigg( Tr_{E^{\prime}}\bigg\{\mathscr{V}_{t}\Big(\boldsymbol{\hat{\rho}}_{S_{0}}\otimes\boldsymbol{\hat{\rho}}^{E,0}\otimes\boldsymbol{\hat{\rho}}^{E^{\prime},0}\Big)\bigg\}\Bigg) = 
\end{equation}
\begin{equation}
\label{eqn:park1}
\Big(\mathscr{U}_{t}\otimes\mathcal{I}_{E^{\prime}}\Big)\circ\Big( \mathcal{I}_{SE}\otimes Tr_{E^{\prime}}\Big)\circ\mathscr{V}_{t}\Big(\boldsymbol{\hat{\rho}}_{S_{0}}\otimes\boldsymbol{\hat{\rho}}^{E,0}\otimes\boldsymbol{\hat{\rho}}^{E^{\prime},0}\Big)
\end{equation}
where $\mathcal{I}_{SE}$ is the identity quantum map acting in the tensor product of spaces, describing the system $S$ and the environment $E$. Since the quantum maps $\mathcal{I}_{SE}\otimes Tr_{E^{\prime}}$ and $\mathscr{U}_{t}\otimes\mathcal{I}_{E^{\prime}}$ act nontrivially on distinct factors of the tensor product, they commute. Hence 
\begin{equation}
\label{eqn:park2}
(\ref{eqn:park1}) = \Big( \mathcal{I}_{SE}\otimes Tr_{E^{\prime}}\Big)\circ\Big(\mathscr{U}_{t}\otimes\mathcal{I}_{E^{\prime}}\Big)\circ\mathscr{V}_{t}\Big(\boldsymbol{\hat{\rho}}_{S_{0}}\otimes\boldsymbol{\hat{\rho}}^{E,0}\otimes\boldsymbol{\hat{\rho}}^{E^{\prime},0}\Big) 
\end{equation}
Furthermore, it is easy to see that the generators of the unitary maps $\mathscr{U}_{t}\otimes\mathcal{I}_{E^{\prime}}$ and $\mathscr{V}_{t}$, namely $\mathbf{\hat{X}}\otimes\mathbf{\hat{B}}\otimes\mathbb{I}_{E^{\prime}}$ and $\mathbf{\hat{X}}\otimes\mathbb{I}_{E}\otimes\mathbf{\hat{B}^{\prime}}$, commute. Therefore, for any $\boldsymbol{\hat{\rho}}_{S_{0}}$ and $\boldsymbol{\hat{\rho}}^{E,0}$,
\begin{equation}
(\ref{eqn:park2}) = \Big( \mathcal{I}_{SE}\otimes Tr_{E^{\prime}}\Big)\circ\mathscr{V}_{t}\circ\Big(\mathscr{U}_{t}\otimes\mathcal{I}_{E^{\prime}}\Big)\Big(\boldsymbol{\hat{\rho}}_{S_{0}}\otimes\boldsymbol{\hat{\rho}}^{E,0}\otimes\boldsymbol{\hat{\rho}}^{E^{\prime},0}\Big)  =   
\end{equation}
\begin{equation}
 \Big( \mathcal{I}_{SE}\otimes Tr_{E^{\prime}}\Big)\circ\mathscr{V}_{t}\Big(\mathscr{U}_{t}\Big(\boldsymbol{\hat{\rho}}_{S_{0}}\otimes\boldsymbol{\hat{\rho}}^{E,0}\Big)\otimes\boldsymbol{\hat{\rho}}^{E^{\prime},0}\Big) =
\end{equation}
\begin{equation}
Tr_{E^{\prime}}\bigg\{\mathscr{V}_{t}\Big(\mathscr{U}_{t}\Big(\boldsymbol{\hat{\rho}}_{S_{0}}\otimes\boldsymbol{\hat{\rho}}^{E,0}\Big)\otimes\boldsymbol{\hat{\rho}}^{E^{\prime},0}\Big)\bigg\} = \big(\mathscr{E}_{t}\otimes\mathcal{I}_{E}\big)\Big(\mathscr{U}_{t}\Big(\boldsymbol{\hat{\rho}}_{S_{0}}\otimes\boldsymbol{\hat{\rho}}^{E,0}\Big)\Big) =
\end{equation}
\begin{equation}
\big(\mathscr{E}_{t}\otimes\mathcal{I}_{E}\big)\circ\mathscr{U}_{t}\Big(\boldsymbol{\hat{\rho}}_{S_{0}}\otimes\boldsymbol{\hat{\rho}}^{E,0}\Big)
\end{equation}
\end{proof}

Lemma \ref{eqn:removingmaps} has  the following corollary:

\begin{Co}[\textbf{An} $\mathscr{E}_{t}$- \textbf{independent estimate}]
\label{eqn:mapsaregoneforthe}
\begin{equation}
\frac{1}{2}\bigg\|\sum_{i}\Big(\mathbf{\hat{P}}_{\Delta_{i,t}}\otimes \mathbb{I}\Big)\boldsymbol{\hat{\rho}}_{t}\Big(\mathbf{\hat{P}}_{\Delta_{i,t}}\otimes \mathbb{I}\Big)- \sum_{i}\Big(\mathbf{\hat{P}}_{\Delta_{i,t}}\otimes \mathbf{\hat{P}}^{E,t}_{i}\Big)\boldsymbol{\hat{\rho}}_{t}\Big(\mathbf{\hat{P}}_{\Delta_{i,t}}\otimes \mathbf{\hat{P}}^{E,t}_{i}\Big)   \bigg\|_{1} \leq 
\end{equation}
\begin{equation}
\label{mapremoved}
\frac{1}{2}\sum_{i}\bar{p}_{i}(t)\bigg\|\mathscr{U}_{t}\big(\boldsymbol{\hat{\rho}}_{S_{i,t}}\otimes \boldsymbol{\hat{\rho}}^{E,0}\big) -\big(\mathbb{I}\otimes \mathbf{\hat{P}}^{E,t}_{i}\big)\mathscr{U}_{t}\big(\boldsymbol{\hat{\rho}}_{S_{i,t}}\otimes \boldsymbol{\hat{\rho}}^{E,0}\big)\big(\mathbb{I}\otimes \mathbf{\hat{P}}^{E,t}_{i}\big)\bigg\|_{1}    
\end{equation}
\end{Co}
\begin{proof}
 \begin{equation}
\bigg\|\sum_{i}\Big(\mathbf{\hat{P}}_{\Delta_{i,t}}\otimes \mathbb{I}\Big)\boldsymbol{\hat{\rho}}_{t}\Big(\mathbf{\hat{P}}_{\Delta_{i,t}}\otimes \mathbb{I}\Big)- \sum_{i}\Big(\mathbf{\hat{P}}_{\Delta_{i,t}}\otimes \mathbf{\hat{P}}^{E,t}_{i}\Big)\boldsymbol{\hat{\rho}}_{t}\Big(\mathbf{\hat{P}}_{\Delta_{i,t}}\otimes \mathbf{\hat{P}}^{E,t}_{i}\Big)   \bigg\|_{1}=
\end{equation}
\begin{equation}
\bigg\|\sum_{i}\bar{p}_{i}(t)\bigg(\mathscr{U}_{t}\big(\mathscr{E}_{t}\big(\boldsymbol{\hat{\rho}}_{S_{i,t}}\big)\otimes \boldsymbol{\hat{\rho}}^{E,0}\big) -\big(\mathbb{I}\otimes \mathbf{\hat{P}}^{E,t}_{i}\big)\mathscr{U}_{t}\big(\mathscr{E}_{t}\big(\boldsymbol{\hat{\rho}}_{S_{i,t}}\big)\otimes \boldsymbol{\hat{\rho}}^{E,0}\big)\big(\mathbb{I}\otimes \mathbf{\hat{P}}^{E_{t}}_{i}\big)\bigg)\bigg\|_{1}\leq
\end{equation}
\begin{equation}
\label{eqn:whereweelevatemap}
\sum_{i}\bar{p}_{i}\bigg\|\mathscr{U}_{t}\big(\mathscr{E}_{t}\big(\boldsymbol{\hat{\rho}}_{S_{i,t}}\big)\otimes \boldsymbol{\hat{\rho}}^{E,0}\big) -\big(\mathbb{I}\otimes \mathbf{\hat{P}}^{E,t}_{i}\big)\mathscr{U}_{t}\big(\mathscr{E}_{t}\big(\boldsymbol{\hat{\rho}}_{S_{i,t}}\big)\otimes \boldsymbol{\hat{\rho}}^{E,0}\big)\big(\mathbb{I}\otimes \mathbf{\hat{P}}^{E,t}_{i}\big)\bigg\|_{1} =
\end{equation}
\begin{equation}
\label{eqn:whereweelevatemap2}
\sum_{i}\bar{p}_{i}\bigg\|\mathscr{E}_{t}\Bigg(\mathscr{U}_{t}\big(\boldsymbol{\hat{\rho}}_{S_{i,t}}\otimes \boldsymbol{\hat{\rho}}^{E,0}\big) -\big(\mathbb{I}\otimes \mathbf{\hat{P}}^{E,t}_{i}\big)\mathscr{U}_{t}\big(\boldsymbol{\hat{\rho}}_{S_{i,t}}\otimes \boldsymbol{\hat{\rho}}^{E,0}\big)\big(\mathbb{I}\otimes \mathbf{\hat{P}}^{E,t}_{i}\big)\Bigg)\bigg\|_{1}\leq
\end{equation}
\begin{equation}
\label{eqn:tobeestimated}
\sum_{i}\bar{p}_{i}\bigg\|\mathscr{U}_{t}\big(\boldsymbol{\hat{\rho}}_{S_{i,t}}\otimes \boldsymbol{\hat{\rho}}^{E,0}\big) -\big(\mathbb{I}\otimes \mathbf{\hat{P}}^{E,t}_{i}\big)\mathscr{U}_{t}\big(\boldsymbol{\hat{\rho}}_{S_{i,t}}\otimes \boldsymbol{\hat{\rho}}^{E,0}\big)\big(\mathbb{I}\otimes \mathbf{\hat{P}}^{E,t}_{i}\big)\bigg\|_{1}
\end{equation} 
where we have used Theorem \ref{eqn:removingmaps} going from (\ref{eqn:whereweelevatemap}) to (\ref{eqn:whereweelevatemap2}), and in going from (\ref{eqn:whereweelevatemap2}) to (\ref{eqn:tobeestimated}) we used the contractivity property of quantum maps (i.e. $\big\|\mathscr{E}\big(\boldsymbol{\hat{\rho}}\big)- \mathscr{E}\big(\boldsymbol{\hat{\sigma}}\big)\big\|_{1}\leq \big\|\boldsymbol{\hat{\rho}}-\boldsymbol{\hat{\sigma}}\big\|_{1}$, see \cite{Nielsen}). 
\end{proof}
For the sake of clarity in the following development, we will now turn to the case where $\boldsymbol{\hat{\rho}}^{E,0}$ and $\boldsymbol{\rho}_{S_{0}}$ are pure states; generalities to the respective mixed-state version of the results presented below are simply attained via linearity and concavity arguments. These assumptions in turn mean that the sate $\boldsymbol{\hat{\rho}}_{S_{i,t}}$ (\ref{eqn:joy}) is also pure.  As such, without the effects of the quantum map $\mathscr{E}_{t}$, the state $\mathscr{U}_{t}\big(\boldsymbol{\hat{\rho}}_{S_{i,t}}\otimes \boldsymbol{\hat{\rho}}^{E,0}\big)$ in (\ref{mapremoved}) is pure. To accentuate the latter we write $\boldsymbol{\hat{\rho}}_{S_{i,t}}=\big|\psi_{S_{i,t}}\big\rangle\big\langle\psi_{S_{i,t}}\big|$, $\boldsymbol{\rho}^{E,0} = \big|\psi_{E,0}\big\rangle\big\langle\psi_{E,0}\big|$ and we use the definition of the map $\mathscr{U}_{t}$ to express it as a conjugation by the unitary operator $\mathbf{\hat{U}}_{t} : = e^{-it\gamma \mathbf{\hat{X}}\otimes \mathbf{\hat{B}}}$.
\begin{equation}
\mathscr{U}_{t}\big(\boldsymbol{\hat{\rho}}_{S_{i,t}}\otimes \boldsymbol{\hat{\rho}}^{E,0}\big)= \mathbf{\hat{U}}_{t}\Big(\big|\psi_{S_{i,t}}\big\rangle\big\langle \psi_{S_{i,t}}\big|\otimes \big|\psi_{E,0}\big\rangle\big\langle \psi_{E,0}\big|\Big)\mathbf{\hat{U}}^{\dagger}_{t} =  
\end{equation}
\begin{equation}
\bigg(\mathbf{\hat{U}}_{t}\Big(\big|\psi_{S_{i,t}}\big\rangle\otimes \big|\psi_{E,0}\big\rangle\Big)\bigg)\bigg(\mathbf{\hat{U}}_{t}\Big(\big|\psi_{S_{i,t}}\big\rangle\otimes \big|\psi_{E,0}\big\rangle\Big)\bigg)^{\dagger}
\end{equation}
It therefore follows that
\begin{equation}
\Big(\mathbb{I}\otimes \mathbf{\hat{P}}^{E,t}_{i}\Big)\mathscr{U}_{t}\big(\boldsymbol{\hat{\rho}}_{S_{i,t}}\otimes \boldsymbol{\hat{\rho}}^{E,0}\big)\Big(\mathbb{I}\otimes \mathbf{\hat{P}}^{E,t}_{i}\Big) = 
\end{equation}
\begin{equation}
\bigg(\mathbb{I}\otimes\mathbf{\hat{P}}^{E,t}_{i}\mathbf{\hat{U}}_{t}\Big(\big|\psi_{S_{i,t}}\big\rangle\otimes \big|\psi_{E,0}\big\rangle\Big)\bigg)\bigg(\mathbb{I}\otimes \mathbf{\hat{P}}^{E,t}_{i}\mathbf{\hat{U}}_{t}\Big(\big|\psi_{S_{i,t}}\big\rangle\otimes \big|\psi_{E,0}\big\rangle\Big)\bigg)^{\dagger}
\end{equation}

We are now ready to present the main theorem of this section.

\begin{definition4}[\textbf{Estimating the Diagonal Terms} (\ref{eqn:uzbek8})]
\label{eqn:themain}
\smaller
\begin{equation}
\min_{PVM}\frac{1}{2}\bigg\|\sum_{i}\Big(\mathbf{\hat{P}}_{\Delta_{i,t}}\otimes \mathbb{I}\Big)\boldsymbol{\hat{\rho}}_{t}\Big(\mathbf{\hat{P}}_{\Delta_{i,t}}\otimes \mathbb{I}\Big)- \frac{1}{\mathscr{N}(t)}\sum_{i}\Big(\mathbf{\hat{P}}_{\Delta_{i,t}}\otimes \mathbf{\hat{P}}^{E,t}_{i}\Big)\boldsymbol{\hat{\rho}}_{t}\Big(\mathbf{\hat{P}}_{\Delta_{i,t}}\otimes \mathbf{\hat{P}}^{E,t}_{i}\Big)   \bigg\|_{1} \leq 
\end{equation}
\begin{equation}
\label{eqn:non-pureb4}
4\min_{PVM}\sqrt{\sum_{i}\bar{p}_{i}(t)\bigg(1-Tr\bigg\{ \mathbf{\hat{P}}^{E,t}_{i}\Lambda_{t,i}\Big(\big|\psi_{E,0}\big\rangle\big\langle\psi_{E,0}\big|\Big)\mathbf{\hat{P}}^{E,t}_{i}\bigg\}\bigg)}
\end{equation}
Here we have defined $\Lambda_{i,t}$ as follows.  
\begin{equation}
\Lambda_{i,t}\big(\boldsymbol{\hat{\rho}}\big):=\int  |\psi_{S_{i,t}}(x)|^{2}\Big(e^{-it\gamma x\mathbf{\hat{B}}}\boldsymbol{\hat{\rho}}e^{it\gamma x\mathbf{\hat{B}}}\Big)dx  
\end{equation}
We have also defined $\big|\psi_{E,0}\big\rangle\big\langle \psi_{E,0}\big|:= \boldsymbol{\hat{\rho}}^{E,0}$. The rest of the notation has already been presented in Lemma \ref{eqn:rewrittingthingstask}.
\normalsize
\end{definition4}
\begin{proof}
First, we will compute the following traces. Recall that $K_{i,t}(x,y):=\psi_{S_{i,t}}(x)\psi_{S_{i,t}}^{*}(y)$
\begin{equation}
\mathscr{N}_{i}(t):= Tr\bigg\{\Big(\mathbb{I}\otimes \mathbf{\hat{P}}^{E,t}_{i}\Big)\mathscr{U}_{t}\big(\boldsymbol{\hat{\rho}}_{S_{i,t}}\otimes \boldsymbol{\hat{\rho}}^{E,0}\big)\Big(\mathbb{I}\otimes \mathbf{\hat{P}}^{E,t}_{i}\Big)\bigg\} =  
\end{equation}
\begin{equation}
\big\langle \psi_{S_{i,t}}\big|\otimes \big\langle \psi_{E,0}\big|\mathbf{\hat{U}}_{t}^{\dagger}\big(\mathbb{I}\otimes \mathbf{\hat{P}}^{E,t}_{i}\big)\mathbf{\hat{U}}_{t}\big|\psi_{S_{i,t}}\big\rangle\otimes\big|\psi_{E,0}\big\rangle = 
\end{equation}

\begin{equation}
\int \big|\psi_{S_{i,t}}(x)\big|^{2}\Big(\big\langle \psi_{E,0}\big|e^{it\gamma y\mathbf{\hat{B}}}\mathbf{\hat{P}}^{E,t}_{i}e^{-it\gamma x\mathbf{\hat{B}}} \big|\psi_{E,0}\big\rangle \Big)dx = 
\end{equation}

\begin{equation}
\int |\psi_{S_{i,t}}(x)|^{2} \bigg(Tr\big\{\mathbf{\hat{P}}^{E,t}_{i}e^{-it\gamma x\mathbf{\hat{B}}}\big|\psi_{E,0}\big\rangle\big\langle \psi_{E,0}\big|e^{it\gamma x\mathbf{\hat{B}}}\mathbf{\hat{P}}^{E,t}_{i}\big\}\bigg) dx =
\end{equation}
\begin{equation}
Tr\Bigg\{\mathbf{\hat{P}}^{E,t}_{i}\bigg(\int |\psi_{S_{i,t}}(x)|^{2} \Big(e^{-it\gamma x\mathbf{\hat{B}}}\big|\psi_{E,0}\big\rangle\big\langle \psi_{E,0}\big|e^{it\gamma x\mathbf{\hat{B}}}\Big) dx\bigg) \mathbf{\hat{P}}^{E,t}_{i}\Bigg\} =
\end{equation}
\begin{equation}
Tr\bigg\{ \mathbf{\hat{P}}^{E,t}_{i}\Lambda_{i,t}\Big(\big|\psi_{E,0}\big\rangle\big\langle \psi_{E,0}\big|\Big)\mathbf{\hat{P}}^{E,t}_{i}\bigg\}
\end{equation}
Now, let us compute the following trace distance using the formula $\frac{1}{2}\Big\|\big|\psi\big\rangle\big\langle \psi\big|-\big|\psi\big\rangle\big\langle\phi\big|\Big\|_{1} = \sqrt{1-\big|\big\langle \phi\big|\psi\big\rangle\big|^{2}}$ (\cite{Nielsen} section 9.2.3). 
\begin{equation}
\frac{1}{2}\bigg\|\mathscr{U}_{t}\big(\boldsymbol{\hat{\rho}}_{S_{i,t}}\otimes \boldsymbol{\hat{\rho}}^{E,0}\big) -\frac{1}{\mathscr{N}_{i}(t)}\big(\mathbb{I}\otimes \mathbf{\hat{P}}^{E,t}_{i}\big)\mathscr{U}_{t}\big(\boldsymbol{\hat{\rho}}_{S_{i,t}}\otimes \boldsymbol{\hat{\rho}}^{E,0}\big)\big(\mathbb{I}\otimes \mathbf{\hat{P}}^{E,t}_{i}\big)\bigg\|_{1} =
\end{equation}
\smaller
\begin{equation}
\frac{1}{2}\Bigg\|\Bigg(\mathbf{\hat{U}}_{t}\Big(\big|\psi_{S_{i,t}}\big\rangle\otimes \big|\psi_{E,0}\big\rangle\Big)\bigg)\bigg(\mathbf{\hat{U}}_{t}\Big(\big|\psi_{S_{i,t}}\big\rangle\otimes \big|\psi_{E,0}\big\rangle\Big)\Bigg)^{\dagger}-
\Bigg(\frac{\mathbb{I}\otimes \mathbf{\hat{P}}^{E,t}_{i}\mathbf{\hat{U}}_{t}\Big(\big|\psi_{S_{i,t}}\big\rangle\otimes \big|\psi_{E,0}\big\rangle\Big)}{\sqrt{\mathscr{N}_{i}(t)}}\Bigg)\Bigg(\frac{\mathbb{I}\otimes \mathbf{\hat{P}}^{E,t}_{i}\mathbf{\hat{U}}_{t}\Big(\big|\psi_{S_{i,t}}\big\rangle\otimes \big|\psi_{E,0}\big\rangle\Big)}{\sqrt{\mathscr{N}_{i}(t)}}\Bigg)^{\dagger}\Bigg\|_{1} =
\end{equation}
\normalsize
\begin{equation}
\sqrt{1-\Bigg|\frac{\big\langle \psi_{S_{i,t}}\big|\otimes \big\langle \psi_{E,0}\big|\mathbf{\hat{U}}^{\dagger}_{t}\Big(\mathbb{I}\otimes \mathbf{\hat{P}}^{E,t}_{i}\Big)\mathbf{\hat{U}}_{t}\big|\psi_{S_{i,t}}\big\rangle\otimes \big|\psi_{E,0}\big\rangle}{\sqrt{\mathscr{N}_{i}(t)}}\Bigg|^{2}} =   
\end{equation}
\begin{equation}
 \sqrt{1-\bigg|\frac{\mathscr{N}_{i}(t)}{\sqrt{\mathscr{N}_{i}(t)}}\bigg|^{2}} =   \sqrt{1-\mathscr{N}_{i}(t)} 
\end{equation}
Recapitulating, we have 
\begin{equation}
\label{first}
\mathscr{N}_{i}(t) = Tr\bigg\{ \mathbf{\hat{P}}^{E,t}_{i}\Lambda_{i,t}\Big(\big|\psi_{E,0}\big\rangle\big\langle \psi_{E,0}\big|\Big)\mathbf{\hat{P}}^{E,t}_{i}\bigg\}
\end{equation}
and
\begin{equation}
\label{second}
\frac{1}{2}\bigg\|\mathscr{U}_{t}\big(\boldsymbol{\hat{\rho}}_{S_{i,t}}\otimes \boldsymbol{\hat{\rho}}^{E,0}\big) -\frac{1}{\mathscr{N}_{i}(t)}\big(\mathbb{I}\otimes \mathbf{\hat{P}}^{E,t}_{i}\big)\mathscr{U}_{t}\big(\boldsymbol{\hat{\rho}}_{S_{i,t}}\otimes \boldsymbol{\hat{\rho}}^{E,0}\big)\big(\mathbb{I}\otimes \mathbf{\hat{P}}^{E,t}_{i}\big)\bigg\|_{1} = \sqrt{1-\mathscr{N}_{i}(t)} 
\end{equation}
An application of Corollary (\ref{eqn:mapsaregoneforthe}), reduces the proof of the Theorem \ref{eqn:themain} to estimating (\ref{mapremoved}). Using (\ref{first}) and (\ref{second}) we get  
\begin{equation}
\frac{1}{2}\sum_{i}\bar{p}_{i}(t)\bigg\|\mathscr{U}_{t}\big(\boldsymbol{\hat{\rho}}_{S_{i,t}}\otimes \boldsymbol{\hat{\rho}}^{E,0}\big) -\big(\mathbb{I}\otimes \mathbf{\hat{P}}^{E,t}_{i}\big)\mathscr{U}_{t}\big(\boldsymbol{\hat{\rho}}_{S_{i,t}}\otimes \boldsymbol{\hat{\rho}}^{E,0}\big)\big(\mathbb{I}\otimes \mathbf{\hat{P}}^{E,t}_{i}\big)\bigg\|_{1}\leq
\end{equation}
\begin{equation}
\frac{1}{2}\sum_{i}\bar{p}_{i}(t)\Bigg[\bigg\|\mathscr{U}_{t}\big(\boldsymbol{\hat{\rho}}_{S_{i,t}}\otimes \boldsymbol{\hat{\rho}}^{E,0}\big) -\frac{1}{\mathscr{N}_{i}(t)}\big(\mathbb{I}\otimes \mathbf{\hat{P}}^{E,t}_{i}\big)\mathscr{U}_{t}\big(\boldsymbol{\hat{\rho}}_{S_{i,t}}\otimes \boldsymbol{\hat{\rho}}^{E,0}\big)\big(\mathbb{I}\otimes \mathbf{\hat{P}}^{E,t}_{i}\big)\bigg\|_{1} + 
\end{equation}
\begin{equation}
\bigg\|\frac{1}{\mathscr{N}_{i}(t)}\big(\mathbb{I}\otimes \mathbf{\hat{P}}^{E,t}_{i}\big)\mathscr{U}_{t}\big(\boldsymbol{\hat{\rho}}_{S_{i,t}}\otimes \boldsymbol{\hat{\rho}}^{E,0}\big)\big(\mathbb{I}\otimes \mathbf{\hat{P}}^{E,t}_{i}\big)-\big(\mathbb{I}\otimes \mathbf{\hat{P}}^{E,t}_{i}\big)\mathscr{U}_{t}\big(\boldsymbol{\hat{\rho}}_{S_{i,t}}\otimes \boldsymbol{\hat{\rho}}^{E,0}\big)\big(\mathbb{I}\otimes \mathbf{\hat{P}}^{E,t}_{i}\big)\bigg\|_{1}\Bigg] =
\end{equation}
\smaller
\begin{equation}
\frac{1}{2}\sum_{i}\bar{p}_{i}(t)\Bigg[\bigg\|\mathscr{U}_{t}\big(\boldsymbol{\hat{\rho}}_{S_{i,t}}\otimes \boldsymbol{\hat{\rho}}^{E,0}\big) -\frac{1}{\mathscr{N}_{i}(t)}\big(\mathbb{I}\otimes \mathbf{\hat{P}}^{E,t}_{i}\big)\mathscr{U}_{t}\big(\boldsymbol{\hat{\rho}}_{S_{i,t}}\otimes \boldsymbol{\hat{\rho}}^{E,0}\big)\big(\mathbb{I}\otimes \mathbf{\hat{P}}^{E,t}_{i}\big)\bigg\|_{1} +  \big|\frac{1}{\mathscr{N}_{i}(t)}-1\big|\bigg\|\big(\mathbb{I}\otimes \mathbf{\hat{P}}^{E,t}_{i}\big)\mathscr{U}_{t}\big(\boldsymbol{\hat{\rho}}_{S_{i,t}}\otimes \boldsymbol{\hat{\rho}}^{E,0}\big)\big(\mathbb{I}\otimes \mathbf{\hat{P}}^{E,t}_{i}\big)\bigg\|_{1}\Bigg]=
\end{equation}
\normalsize
\begin{equation}
\frac{1}{2}\sum_{i}\bar{p}_{i}(t)\Bigg[\bigg\|\mathscr{U}_{t}\big(\boldsymbol{\hat{\rho}}_{S_{i,t}}\otimes \boldsymbol{\hat{\rho}}^{E,0}\big) -\frac{1}{\mathscr{N}_{i}(t)}\big(\mathbb{I}\otimes \mathbf{\hat{P}}^{E,t}_{i}\big)\mathscr{U}_{t}\big(\boldsymbol{\hat{\rho}}_{S_{i,t}}\otimes \boldsymbol{\hat{\rho}}^{E,0}\big)\big(\mathbb{I}\otimes \mathbf{\hat{P}}^{E,t}_{i}\big)\bigg\|_{1}  +  \big|\frac{1}{\mathscr{N}_{i}(t)}-1\big|\mathscr{N}_{i}(t)\Bigg] =  
\end{equation}
\begin{equation}
\label{eqn:blume}
\frac{1}{2}\sum_{i}\bar{p}_{i}(t)\Bigg[2\sqrt{1-\mathscr{N}_{i}(t)}+1-\mathscr{N}_{i}(t)\Bigg]\leq
\sum_{i}\bar{p}_{i}(t)\Bigg[2\sqrt{1-\mathscr{N}_{i}(t)}\Bigg] =
\end{equation}
\begin{equation}
\label{eqn:Jensens}
2\sum_{i}\bar{p}_{i}(t)\sqrt{1-\mathscr{N}_{i}(t)}\leq 2\sqrt{\sum_{i}\bar{p}_{i}(t)\big(1-\mathscr{N}_{i}(t)\big)}=
\end{equation}
\begin{equation}
\label{eqn:postjens}
2\sqrt{\sum_{i}\bar{p}_{i}(t)\bigg(1-Tr\bigg\{ \mathbf{\hat{P}}^{E,t}_{i}\Lambda_{i,t}\Big(\big|\psi_{E,0}\big\rangle\big\langle\psi_{E,0}\big|\Big)\mathbf{\hat{P}}^{E,t}_{i}\bigg\}\bigg)}
\end{equation}
Here we have employed Jensen's inequality for concave functions in going from (\ref{eqn:Jensens}) to (\ref{eqn:postjens}).

By virtue of Corollary \ref{eqn:mapsaregoneforthe} we therefore have
\begin{equation}
\frac{1}{2}\bigg\|\sum_{i}\Big(\mathbf{\hat{P}}_{\Delta_{i,t}}\otimes \mathbb{I}\Big)\boldsymbol{\hat{\rho}}_{t}\Big(\mathbf{\hat{P}}_{\Delta_{i,t}}\otimes \mathbb{I}\Big)- \sum_{i}\Big(\mathbf{\hat{P}}_{\Delta_{i,t}}\otimes \mathbf{\hat{P}}^{E,t}_{i}\Big)\boldsymbol{\hat{\rho}}_{t}\Big(\mathbf{\hat{P}}_{\Delta_{i,t}}\otimes \mathbf{\hat{P}}^{E,t}_{i}\Big)   \bigg\|_{1} \leq 
\end{equation}
\begin{equation}
2\sqrt{\sum_{i}\bar{p}_{i}(t)\bigg(1-Tr\bigg\{ \mathbf{\hat{P}}^{E,t}_{i}\Lambda_{i,t}\Big(\big|\psi_{E,0}\big\rangle\big\langle\psi_{E,0}\big|\Big)\mathbf{\hat{P}}^{E,t}_{i}\bigg\}\bigg)} 
\end{equation}
Finally, a simple application of Lemma \ref{eqn:reversetin} leads to 
\begin{equation}
\label{spaceso}
\frac{1}{2}\bigg\|\sum_{i}\Big(\mathbf{\hat{P}}_{\Delta_{i,t}}\otimes \mathbb{I}\Big)\boldsymbol{\hat{\rho}}_{t}\Big(\mathbf{\hat{P}}_{\Delta_{i,t}}\otimes \mathbb{I}\Big)- \frac{1}{\mathscr{N}(t)}\sum_{i}\Big(\mathbf{\hat{P}}_{\Delta_{i,t}}\otimes \mathbf{\hat{P}}^{E,t}_{i}\Big)\boldsymbol{\hat{\rho}}_{t}\Big(\mathbf{\hat{P}}_{\Delta_{i,t}}\otimes \mathbf{\hat{P}}^{E,t}_{i}\Big)   \bigg\|_{1} \leq 
\end{equation}
\begin{equation}
\label{spacesod}
4\sqrt{\sum_{i}\bar{p}_{i}(t)\bigg(1-Tr\bigg\{ \mathbf{\hat{P}}^{E,t}_{i}\Lambda_{i,t}\Big(\big|\psi_{E,0}\big\rangle\big\langle\psi_{E,0}\big|\Big)\mathbf{\hat{P}}^{E,t}_{i}\bigg\}\bigg)} 
\end{equation}
Taking the minimum over all PVM acting on the environmental degrees of freedom on both sides of inequality (\ref{spaceso})-(\ref{spacesod}) we get the result we set out to prove. 
\end{proof}

For the cases where instead of $\boldsymbol{\hat{\rho}}^{E,0}=\big|\psi_{E,0}\big\rangle\big\langle \psi_{E,0}\big|$ we have a mixed state $\boldsymbol{\hat{\rho}}^{E,0} = \sum_{j}\mu_{j}\big|\psi_{j,0}\big\rangle\big\langle \psi_{j,0}\big|$ it can easily be shown that Theorem \ref{eqn:themain} still holds; replacing $\big|\psi_{E,0}\big\rangle\big\langle \psi_{E,0}\big|$ with $\sum_{j}\mu_{j}\big|\psi_{j,0}\big\rangle\big\langle \psi_{j,0}\big|$. The key things to notice are that we may pull out the sum over $j$ and the $\mu_{j}$ out of the norm in (\ref{eqn:non-pureb4}) before doing any estimates using triangle inequality. Ignoring the $\mu_{j}$ and the sum over $j$ we can arrive the same result (\ref{eqn:blume}) now for a state $\big|\psi_{j,0}\big\rangle\big\langle \psi_{j,0}\big|$ rather than the state $\big|\psi_{E,0}\big\rangle\big\langle \psi_{E,0}\big|$. Namely,
\begin{equation}
\label{eqn:jenjen}
2\sum_{j}\mu_{j}\sum_{i}\bar{p}_{i}(t)\sqrt{1-\mathscr{N}_{i,j}(t)} = 2\sum_{j}\sum_{i}\mu_{j}\bar{p}_{i}(t)\sqrt{1-\mathscr{N}_{i,j}(t)}
\end{equation}
where now 
\begin{equation}
\mathscr{N}_{i,j}(t) = Tr\bigg\{ \mathbf{\hat{P}}^{E,t}_{i}\Lambda_{i,t}\Big(\big|\psi_{j,0}\big\rangle\big\langle \psi_{j,0}\big|\Big)\mathbf{\hat{P}}^{E,t}_{i}\bigg\}
\end{equation}
However, the $\mu_{j}\bar{p}_{i}(t)$ form a probability distribution, hence we may apply Jensen's inequality for concave functions to the right-hand side of (\ref{eqn:jenjen}). 

\begin{equation}
2\sum_{j}\sum_{i}\mu_{j}\bar{p}_{i}(t)\sqrt{1-\mathscr{N}_{i,j}(t)}\leq 2\sqrt{\sum_{j}\sum_{i}\mu_{j}\bar{p}_{i}(t)\big(1-\mathscr{N}_{i,j}(t)\big)} = 2\sqrt{\sum_{i}\bar{p}_{i}(t)\big(1-\sum_{j}\mu_{j}\mathscr{N}_{i,j}(t)\big)}=
\end{equation}
\begin{equation}
2\sqrt{\sum_{i}\bar{p}_{i}(t)\big(1- Tr\bigg\{ \mathbf{\hat{P}}^{E,t}_{i}\Lambda_{i,t}\Big(\sum_{j}\mu_{j}\big|\psi_{j,0}\big\rangle\big\langle \psi_{j,0}\big|\Big)\mathbf{\hat{P}}^{E,t}_{i}\bigg\}\big)}= 
\end{equation}

\begin{equation}
2\sqrt{\sum_{i}\bar{p}_{i}(t)\big(1- Tr\bigg\{ \mathbf{\hat{P}}^{E,t}_{i}\Lambda_{i,t}\Big(\boldsymbol{\hat{\rho}}^{E,0}\Big)\mathbf{\hat{P}}^{E,t}_{i}\bigg\}\big)} =2\sqrt{\sum_{i}\bar{p}_{i}(t)\big(1-\mathscr{N}_{i}(t)\big)}    
\end{equation}
and so we have our result. We present the latter below formally as a corollary.  
\begin{Co}[\textbf{Estimating the Diagonal Terms} (\ref{eqn:uzbek8}), with $\boldsymbol{\hat{\rho}}^{E,0}$ a mixture.]
\label{eqn:joy2}
\smaller
\begin{equation}
\min_{PVM}\frac{1}{2}\bigg\|\sum_{i}\Big(\mathbf{\hat{P}}_{\Delta_{i,t}}\otimes \mathbb{I}\Big)\boldsymbol{\hat{\rho}}_{t}\Big(\mathbf{\hat{P}}_{\Delta_{i,t}}\otimes \mathbb{I}\Big)- \frac{1}{\mathscr{N}(t)}\sum_{i}\Big(\mathbf{\hat{P}}_{\Delta_{i,t}}\otimes \mathbf{\hat{P}}^{E,t}_{i}\Big)\boldsymbol{\hat{\rho}}_{t}\Big(\mathbf{\hat{P}}_{\Delta_{i,t}}\otimes \mathbf{\hat{P}}^{E,t}_{i}\Big)   \bigg\|_{1} \leq 
\end{equation}
\begin{equation}
\label{eqn:non-pureb4}
4\min_{PVM}\sqrt{\sum_{i}\bar{p}_{i}(t)\bigg(1-Tr\bigg\{ \mathbf{\hat{P}}^{E,t}_{i}\Lambda_{t,i}\Big(\boldsymbol{\hat{\rho}}^{E,0}\Big)\mathbf{\hat{P}}^{E,t}_{i}\bigg\}\bigg)}
\end{equation}
\normalsize
Here we have defined 
\smaller
$\Lambda_{i,t}$ as follows.  
\begin{equation}
\label{eqn:joy3}
\Lambda_{i,t}\big(\boldsymbol{\hat{\rho}}\big):=\int  |\psi_{S_{i,t}}(x)|^{2}\Big(e^{-it\gamma x\mathbf{\hat{B}}}\boldsymbol{\hat{\rho}}e^{it\gamma x\mathbf{\hat{B}}}\Big)dx  
\end{equation}
\normalsize
The rest of the notation has already been presented in Lemma \ref{eqn:rewrittingthingstask}.
\end{Co}
\normalsize
Using similar linearity and concavity arguments we may extend to further generalities for the case where $\boldsymbol{\hat{\rho}}_{S_{0}}$ is a mixed state. In such a case Corrolary \ref{eqn:joy2} remains virtually unchanged except for the fact that the quantum map $\Lambda_{i,t}$ (\ref{eqn:joy3}) is now replaced by one which involves a sum $\sum_{n}q_{n}|\psi_{S^{n}_{i,t}}(x)|^{2}$ (see \ref{eqn:joy4}) in the integrand instead of a single term. We will not present the latter formally as a result since the remainder of the paper is developed for a pure state $\boldsymbol{\hat{\rho}}_{S_{0}}$.

\section{Gneralizing to $N_{E}$ Environments}
As mentioned earlier, in the discussion spanning equations (\ref{eqn:intham2}) through (\ref{eqn:evo12}), Theorem \ref{eqn:themain} may be easily generalized to the setting where $N_{E}$ environments are present and $M_{E}$ environments have been traced out. We present this without proof since the steps are analogous to all of the steps involved in proving Theorem \ref{eqn:themain} and the generalization Corrolary \ref{eqn:joy2}.

\begin{definition4}[\textbf{Estimating the Diagonal Terms} (\ref{eqn:uzbek8}) \textbf{for} $N_{E}$ \textbf{Environments}]
\label{eqn:themain2}
\smaller
\begin{equation}
\min_{PVM}\frac{1}{2}\bigg\|\sum_{i}\Big(\mathbf{\hat{P}}_{\Delta_{i,t}}\otimes \mathbb{I}\Big)\boldsymbol{\hat{\rho}}_{t}\Big(\mathbf{\hat{P}}_{\Delta_{i,t}}\otimes \mathbb{I}\Big)- \frac{1}{\mathscr{N}(t)}\sum_{i}\Big(\mathbf{\hat{P}}_{\Delta_{i,t}}\otimes\bigotimes_{k=1}^{N_{E}} \mathbf{\hat{P}}^{E^{k},t}_{i}\Big)\boldsymbol{\hat{\rho}}_{t}\Big(\mathbf{\hat{P}}_{\Delta_{i,t}}\otimes\bigotimes_{k=1}^{N_{E}} \mathbf{\hat{P}}^{E^{k},t}_{i}\Big)   \bigg\|_{1} \leq 
\end{equation}
\begin{equation}
\label{eqn:non-pureb}
\min_{PVM}4\sqrt{\sum_{i}\bar{p}_{i}(t)\Bigg(1-Tr\Bigg\{ \bigotimes_{k=1}^{N_{E}}\mathbf{\hat{P}}^{E^{k},t}_{i}\Lambda_{i,t}\Bigg(\bigotimes_{k=1}^{N_{E}}\boldsymbol{\hat{\rho}}^{E^{k},0}\Bigg)\bigotimes_{k=1}^{N_{E}}\mathbf{\hat{P}}^{E^{k},t}_{i}\Bigg\}\Bigg)}
\end{equation}
Here, we have defined $\Lambda_{i,t}$ as follows.  
\begin{equation}
\Lambda_{i,t}\big(\boldsymbol{\hat{\rho}}\big):=\int  |\psi_{S_{i,t}}(x)|^{2}\Big(e^{-it x\sum_{k=1}^{N_{E}}g_{k}\mathbf{\hat{B}}_{k}}\boldsymbol{\hat{\rho}}e^{it x\sum_{k=1}^{N_{E}}g_{k}\mathbf{\hat{B}}_{k}}\Big)dx  
\end{equation}
Each $\mathbf{\hat{B}}_{k}$ acts on the respective Hilbert space. 
\normalsize
\end{definition4}

In the special case where all of the $\boldsymbol{\hat{\rho}}^{E^{k},0}$ are identical and pure,  and all of the $g_{k}\mathbf{\hat{B}}_{k}$ are identical, Theorem \ref{eqn:themain2} takes the following simpler form.   
\begin{Co}[\textbf{Estimating the Diagonal Terms} (\ref{eqn:uzbek8}) \textbf{for} $N_{E}$ \textbf{identical Environments with identical} $\mathbf{\hat{B}}_{k}$ \textbf{and identical} $g_{k}$]
\label{eqn:themain3}
\smaller
\begin{equation}
\min_{PVM}\frac{1}{2}\bigg\|\sum_{i}\Big(\mathbf{\hat{P}}_{\Delta_{i,t}}\otimes \mathbb{I}\Big)\boldsymbol{\hat{\rho}}_{t}\Big(\mathbf{\hat{P}}_{\Delta_{i,t}}\otimes \mathbb{I}\Big)- \frac{1}{\mathscr{N}(t)}\sum_{i}\Big(\mathbf{\hat{P}}_{\Delta_{i,t}}\otimes\bigotimes_{k=1}^{N_{E}} \mathbf{\hat{P}}^{E^{k},t}_{i}\Big)\boldsymbol{\hat{\rho}}_{t}\Big(\mathbf{\hat{P}}_{\Delta_{i,t}}\otimes\bigotimes_{k=1}^{N_{E}} \mathbf{\hat{P}}^{E^{k},t}_{i}\Big)   \bigg\|_{1} \leq 
\end{equation}
\begin{equation}
\label{eqn:non-pureb2}
\min_{PVM}4\sqrt{\sum_{i}\bar{p}_{i}(t)\Bigg(1-  \int|\psi_{S_{i,t}}(x)|^{2} \big\langle\psi_{E,t}(x)\big|\mathbf{\hat{P}}^{E,t}_{i}\big|\psi_{E,t}(x)\big\rangle^{N_{E}}dx\Bigg)}
\end{equation}
where
\begin{equation}
\big|\psi_{E,t}(x)\big\rangle:=e^{-itxg\mathbf{\hat{B}}}\big|\psi_{E,0}\big\rangle     
\end{equation}
\normalsize
\end{Co}

Theorems \ref{eqn:themain} and \ref{eqn:themain2}, and Corollaries \ref{eqn:joy2} and  \ref{eqn:themain3} are tools that may aid in the estimation of the diagonal term (\ref{eqn:uzbek8}). Their drawback is that they require one to find an approximately optimal PVM acting on the environmental degrees of freedom.  This is in contrast to the discrete variables case, where a PVM-independent bound exists (see Theorem 4 of \cite{AA}). It is important to note that, since the density operators $\Lambda_{t,i}\big(\boldsymbol{\hat{\rho}}^{E,0}\big)$ are not pure,  Theorem  4 of \cite{AA} cannot be applied in this case.  Nevertheless, it may be applicable in an approximate manner when the $\Lambda_{i,t}\big(\boldsymbol{\hat{\rho}}^{E,0}\big)$ are approximately pure. To see this, let us consider the term under the square root of (\ref{eqn:non-pureb}). Defining $\boldsymbol{\hat{\rho}}^{E,t}_{x_{i}}:= e^{-it\gamma x_{i}\mathbf{\hat{B}}}\boldsymbol{\hat{\rho}}^{E,0}e^{it\gamma x_{i}\mathbf{\hat{B}}} $, where $x_{i}:= \int x|\psi_{S_{i,t}}(x)|^{2}dx$, we have 
\begin{equation}
1- Tr\big\{ \mathbf{\hat{P}}^{E,t}_{i}\Lambda_{i,t}\big(\boldsymbol{\hat{\rho}}^{E,0}\big)\mathbf{\hat{P}}^{E,t}_{i} \big\} \leq
\end{equation}
\begin{equation}
\bigg\|\Lambda_{i,t}\big(\boldsymbol{\hat{\rho}}^{E,0}\big)-\mathbf{\hat{P}}^{E,t}_{i}\Lambda_{i,t}\big(\boldsymbol{\hat{\rho}}^{E,0}\big)\mathbf{\hat{P}}^{E,t}_{i}\bigg\|_{1} = 
\end{equation}
\begin{equation}
\bigg\|\Lambda_{i,t}\big(\boldsymbol{\hat{\rho}}^{E,0}\big)-\boldsymbol{\hat{\rho}}^{E,t}_{x_{i}}+\boldsymbol{\hat{\rho}}^{E,t}_{x_{i}}-\mathbf{\hat{P}}^{E,t}_{i}\boldsymbol{\hat{\rho}}^{E,t}_{x_{i}}\mathbf{\hat{P}}^{E,t}_{i}+\mathbf{\hat{P}}^{E,t}_{i}\boldsymbol{\hat{\rho}}^{E,t}_{x_{i}}\mathbf{\hat{P}}^{E,t}_{i}-\mathbf{\hat{P}}^{E,t}_{i}\Lambda_{i,t}\big(\boldsymbol{\hat{\rho}}^{E,0}\big)\mathbf{\hat{P}}^{E,t}_{i}\bigg\|_{1}\leq
\end{equation}
\begin{equation}
\bigg\|\Lambda_{i,t}\big(\boldsymbol{\hat{\rho}}^{E,0}\big)-\boldsymbol{\hat{\rho}}^{E,t}_{x_{i}}\bigg\|_{1}+\bigg\|\boldsymbol{\hat{\rho}}^{E,t}_{x_{i}}-\mathbf{\hat{P}}^{E,t}_{i}\boldsymbol{\hat{\rho}}^{E,t}_{x_{i}}\mathbf{\hat{P}}^{E,t}_{i}\bigg\|_{1}+\bigg\|\mathbf{\hat{P}}^{E,t}_{i}\boldsymbol{\hat{\rho}}^{E,t}_{x_{i}}\mathbf{\hat{P}}^{E,t}_{i}-\mathbf{\hat{P}}^{E,t}_{i}\Lambda_{i,t}\big(\boldsymbol{\hat{\rho}}^{E,0}\big)\mathbf{\hat{P}}^{E,t}_{i}\bigg\|_{1}\leq
\end{equation}
\begin{equation}
\bigg\|\Lambda_{i,t}\big(\boldsymbol{\hat{\rho}}^{E,0}\big)-\boldsymbol{\hat{\rho}}^{E,t}_{x_{i}}\bigg\|_{1}+\bigg\|\boldsymbol{\hat{\rho}}^{E,t}_{x_{i}}-\mathbf{\hat{P}}^{E,t}_{i}\boldsymbol{\hat{\rho}}^{E,t}_{x_{i}}\mathbf{\hat{P}}^{E,t}_{i}\bigg\|_{1}+\bigg\|\Lambda_{i,t}\big(\boldsymbol{\hat{\rho}}^{E,0}\big)-\boldsymbol{\hat{\rho}}^{E,t}_{x_{i}}\bigg\|_{1}= 
\end{equation}
\begin{equation}
2\bigg\|\Lambda_{i,t}\big(\boldsymbol{\hat{\rho}}^{E,0}\big)-\boldsymbol{\hat{\rho}}^{E,t}_{x_{i}}\bigg\|_{1}+\bigg\|\boldsymbol{\hat{\rho}}^{E,t}_{x_{i}}-\mathbf{\hat{P}}^{E,t}_{i}\boldsymbol{\hat{\rho}}^{E,t}_{x_{i}}\mathbf{\hat{P}}^{E,t}_{i}\bigg\|_{1} 
\end{equation}
This allows to estimate (\ref{eqn:non-pureb}) as follows:
\begin{equation}
4\min_{PVM} \sqrt{\sum_{i}\bar{p}_{i}(t)\bigg(1- Tr\big\{ \mathbf{\hat{P}}^{E,t}_{i}\Lambda_{i,t}\big(\boldsymbol{\hat{\rho}}^{E,0}\big)\mathbf{\hat{P}}^{E,t}_{i} \big\}\bigg)} \leq
\end{equation}
\begin{equation}
4\min_{PVM}\sqrt{\sum_{i}\bar{p}_{i}(t)\bigg(2\bigg\|\Lambda_{i,t}\big(\boldsymbol{\hat{\rho}}^{E,0}\big)-\boldsymbol{\hat{\rho}}^{E,t}_{x_{i}}\bigg\|_{1}+\bigg\|\boldsymbol{\hat{\rho}}^{E,t}_{x_{i}}-\mathbf{\hat{P}}^{E,t}_{i}\boldsymbol{\hat{\rho}}^{E,t}_{x_{i}}\mathbf{\hat{P}}^{E,t}_{i}\bigg\|_{1}\bigg)}\leq
\end{equation}
\begin{equation}
4\sqrt{2\sum_{i}\bar{p}_{i}(t)\bigg\|\Lambda_{i,t}\big(\boldsymbol{\hat{\rho}}^{E,0}\big)-\boldsymbol{\hat{\rho}}^{E,t}_{x_{i}}\bigg\|_{1}}+4\min_{PVM}\sqrt{\sum_{i}\bar{p}_{i}(t)\bigg\|\boldsymbol{\hat{\rho}}^{E,t}_{x_{i}}-\mathbf{\hat{P}}^{E,t}_{i}\boldsymbol{\hat{\rho}}^{E,t}_{x_{i}}\mathbf{\hat{P}}^{E,t}_{i}\bigg\|_{1}}
\end{equation}
We state the last result as a lemma. 
\begin{definition2}[\textbf{Diagonal terms 
of the SBS problem for continuous variables, further estimates}]
\label{eqn:diagterms}
\begin{equation}
\min_{PVM}\frac{1}{2}\bigg\|\sum_{i}\Big(\mathbf{\hat{P}}_{\Delta_{i,t}}\otimes \mathbb{I}\Big)\boldsymbol{\hat{\rho}}_{t}\Big(\mathbf{\hat{P}}_{\Delta_{i,t}}\otimes \mathbb{I}\Big)- \frac{1}{\mathscr{N}(t)}\sum_{i}\Big(\mathbf{\hat{P}}_{\Delta_{i,t}}\otimes \mathbf{\hat{P}}^{E,t}_{i}\Big)\boldsymbol{\hat{\rho}}_{t}\Big(\mathbf{\hat{P}}_{\Delta_{i,t}}\otimes \mathbf{\hat{P}}^{E,t}_{i}\Big)   \bigg\|_{1} \leq
\end{equation}
\begin{equation}
4\sqrt{2\sum_{i}\bar{p}_{i}(t)\bigg\|\Lambda_{i,t}\big(\boldsymbol{\hat{\rho}}^{E,0}\big)-\boldsymbol{\hat{\rho}}^{E,t}_{x_{i}}\bigg\|_{1}}+4\min_{PVM}\sqrt{\sum_{i}\bar{p}_{i}(t)\bigg\|\boldsymbol{\hat{\rho}}^{E,t}_{x_{i}}-\mathbf{\hat{P}}^{E,t}_{i}\boldsymbol{\hat{\rho}}^{E,t}_{x_{i}}\mathbf{\hat{P}}^{E,t}_{i}\bigg\|_{1}}
\end{equation}
\end{definition2}
This can be easily extended to the case where we have more than one environmental degree of freedom. In such a case, Lemma \ref{eqn:diagterms} becomes the following.  

\begin{Co}[\textbf{Diagonal terms for continuous variables} in the case of $N_{E}$ \textbf{environments; further estimates}]
\label{eqn:tensdiagg}
\begin{equation}
\min_{PVM}\frac{1}{2}\bigg\|\sum_{i}\Big(\mathbf{\hat{P}}_{\Delta_{i,t}}\otimes \mathbb{I}\Big)\boldsymbol{\hat{\rho}}_{t}\Big(\mathbf{\hat{P}}_{\Delta_{i,t}}\otimes \mathbb{I}\Big)- \frac{1}{\mathscr{N}(t)}\sum_{i}\Big(\mathbf{\hat{P}}_{\Delta_{i,t}}\otimes\bigotimes_{k=1}^{N_{E}} \mathbf{\hat{P}}^{E^{k},{t}}_{i}\Big)\boldsymbol{\hat{\rho}}_{t}\Big(\mathbf{\hat{P}}_{\Delta_{i,t}}\otimes\bigotimes_{k=1}^{N_{E}} \mathbf{\hat{P}}^{E^{k},{t}}_{i}\Big) \bigg\|_{1}\leq
\end{equation}
\begin{equation}
4\sqrt{2\sum_{i}\bar{p}_{i}(t)\Bigg\|\Lambda_{i,t}\Bigg(\bigotimes_{k=1}^{N_{E}}\boldsymbol{\hat{\rho}}^{E^{k},0}\Bigg)-\bigotimes_{k=1}^{N_{E}}\boldsymbol{\hat{\rho}}^{E^{k},t}_{x_{i}}\Bigg\|_{1}}+4\min_{PVM} \sqrt{\sum_{i}\bar{p}_{i}(t)\Bigg\|\bigotimes_{k=1}^{N_{E}}\boldsymbol{\hat{\rho}}^{E^{k},t}_{x_{i}}-\bigotimes_{k=1}^{N_{E}}\mathbf{\hat{P}}^{E^{k},t}_{i}\boldsymbol{\hat{\rho}}^{E^{k},t}_{x_{i}}\mathbf{\hat{P}}^{E^{k},t}_{i}\Bigg\|_{1}}
\end{equation}
\end{Co}
In the sequel, the following simple inequality will be useful. 
\begin{definition2}[\textbf{Telescopic inequality} \cite{JKthree}]
Let $\mathbf{\hat{A}}^{k}$ and $\mathbf{\hat{B}}^{k}$ be trace class operators for all $k$. Then,
\label{eqn:telescoping}
\begin{equation}
   \Big\| \bigotimes_{k=1}^{N}\mathbf{\hat{A}}^{k}-\bigotimes_{k=1}^{N}\mathbf{\hat{B}}^{k}  \Big\|_{1} \leq
\end{equation}
\begin{equation}
\sum_{j=1}^{N}\bigg(\prod_{k=1}^{j-1}\big\|\mathbf{\hat{A}}^{k}\big\|_{1}\bigg)\times \big\|\mathbf{\hat{A}}^{j}-\mathbf{\hat{B}}^{j}\big\|_{1}\times \bigg(\prod_{k=j+1}^{N}\big\|\mathbf{\hat{B}}^{k} 
 \big\|_{1}\bigg)
\end{equation}
\end{definition2}
Using Corollary \ref{eqn:tensdiagg} and Lemma \ref{eqn:telescoping} we obtain the following useful corollary. 

\begin{Co}[\textbf{Further estimates}]
\label{eqn:thecorrolaryfordiags}
\begin{equation}
\frac{1}{2}\min_{PVM}\bigg\|\sum_{i}\Big(\mathbf{\hat{P}}_{\Delta_{i,t}}\otimes \mathbb{I}\Big)\boldsymbol{\hat{\rho}}_{t}\Big(\mathbf{\hat{P}}_{\Delta_{i,t}}\otimes \mathbb{I}\Big)- \frac{1}{\mathscr{N}(t)}\sum_{i}\Big(\mathbf{\hat{P}}_{\Delta_{i,t}}\otimes\bigotimes_{k=1}^{N_{E}} \mathbf{\hat{P}}^{E^{k},t}_{i}\Big)\boldsymbol{\hat{\rho}}_{t}\Big(\mathbf{\hat{P}}_{\Delta_{i,t}}\otimes\bigotimes_{k=1}^{N_{E}} \mathbf{\hat{P}}^{E^{k},t}_{i}\Big)\bigg\|_{1}\leq
\end{equation}
\begin{equation}
\label{eqn:further}
4\sqrt{2\sum_{i}\bar{p}_{i}(t)\sum_{k=1}^{N_{E}}\int|\psi_{S_{i,t}}(x)|^{2}\Big\|\boldsymbol{\hat{\rho}}^{E^{k},{t}}_{x}-\boldsymbol{\hat{\rho}}^{E^{k},t}_{x_{i}}\Big\|_{1}dx}+4\min_{PVM} \sqrt{\sum_{i}\bar{p}_{i}(t)\sum_{k=1}^{N_{E}}\Big\|\boldsymbol{\hat{\rho}}^{E^{k},t}_{x_{i}}-\mathbf{\hat{P}}^{E^{k},t}_{i}\boldsymbol{\hat{\rho}}^{E^{k},t}_{x_{i}}\mathbf{\hat{P}}^{E^{k},t}_{i}\Big\|_{1}}
\end{equation}
\end{Co}

\begin{proof}
First note that 
\begin{equation}
\Bigg\|\Lambda_{i,t}\Bigg(\bigotimes_{k=1}^{N_{E}}\boldsymbol{\hat{\rho}}^{E^{k},0}\Bigg)-\bigotimes_{k=1}^{N_{E}}\boldsymbol{\hat{\rho}}^{E^{k},t}_{x_{i}}\Bigg\|_{1}=\Bigg\|\int|\psi_{S_{i,t}}(x)|^{2}\bigg(\bigotimes_{k=1}^{N_{E}}\boldsymbol{\hat{\rho}}^{E^{k},t}_{x}\bigg)dx-\bigotimes_{k=1}^{N_{E}}\boldsymbol{\hat{\rho}}^{E^{k},t}_{x_{i}}\Bigg\|_{1}=
\end{equation}
\begin{equation}
 \Bigg\|\int|\psi_{S_{i,t}}(x)|^{2}\bigg(\bigotimes_{k=1}^{N_{E}}\boldsymbol{\hat{\rho}}^{E^{k},t}_{x}-\bigotimes_{k=1}^{N_{E}}\boldsymbol{\hat{\rho}}^{E^{k},t}_{x_{i}}\bigg)dx\Bigg\|_{1}\leq
\end{equation}
\begin{equation}
\int|\psi_{S_{i,t}}(x)|^{2}\Bigg\|\bigg(\bigotimes_{k=1}^{N_{E}}\boldsymbol{\hat{\rho}}^{E^{k},t}_{x}-\bigotimes_{k=1}^{N_{E}}\boldsymbol{\hat{\rho}}^{E^{k},t}_{x_{i}}\bigg)\Bigg\|_{1}dx.
\end{equation}
Using the latter, the proof follows directly from Lemma \ref{eqn:telescoping} and Theorem \ref{eqn:diagterms} by noting that 
\begin{equation}
\big\|\boldsymbol{\hat{\rho}}^{E^{k},t}_{x}\big\|_{1} = 1
\end{equation}
for all $t$, $k$ and $x$. 
\end{proof}
If the first term of (\ref{eqn:further}) is small, then we may benefit from the use of Theorem 4 \cite{AA} in estimating the second term of (\ref{eqn:further}) in a PVM-independent way.

\begin{section}{Conclusion}
The results of this paper provide tools for the estimation of the diagonal and off-diagonal terms first seen in (\ref{eqn:uzbek8}) and (\ref{eqn:uzbek9}). Estimating these terms is necessary to understand the asymptotic behavior of the optimization (\ref{eqn:uzbek}) which is the primary focus of this paper. Assuming that the partial tracing that takes place in section \ref{eqn:ptrace}, via implementation of Lemma \ref{eqn:partialtrace} in appendix \ref{thelemmainthispape}, leads to a kernel $\Gamma(t,x,y)$ which, along with its derivative for with respect to $y$, decays to zero as $t\rightarrow \infty$ we may apply Theorem \ref{eqn:theoremkupsch} to prove that the off-diagonal terms decay to zero as $t\rightarrow \infty$. A word of warning is now merited, the attentive reader may have noticed that Theorem \ref{eqn:theoremkupsch} is not useful in the case where the kernel of the density operator $\boldsymbol{\hat{\rho}}_{S_{0}}$ is not supported on a set of finite measure. This is because one would then up having to sum over an infinite amount of terms of the form (\ref{eqn:kupschkupsch}) producing a term that will not decay to zero as $t\rightarrow \infty$ even if $\Gamma(t,x,y)$ does decay to zero as $t\rightarrow \infty$ for all $x,y\in\mathbb{R}$. The way to work around this is to just assume that the kernel of $\boldsymbol{\hat{\rho}}_{S_{0}}$ is supported over a set of finite measures; future work could perhaps focus on strengthening the bound afforded by Theorem \ref{eqn:theoremkupsch} amongst other things.  

Now, Assuming that the off-diagonal terms decay to zero; for our setting, it suffices to assume that the kernel of $\boldsymbol{\hat{\rho}}_{S_{0}}$ is supported on a set of finite measure and that the states $\boldsymbol{\hat{\rho}}^{E^{k},0}$ are supported on the \emph{Rajchman} subspace of $\mathbf{\hat{B}}_{k}$ (see section 7 of \cite{AA}). In \cite{AA} analogous results to those produced in this work are produced for the case of discrete variables, i.e. $\mathbf{\hat{X}}$ is assumed to have purely discrete spectrum, but there the bounds where all independent of the minimization over all PVM. As such, arguments about the convergence of the corresponding (\ref{eqn:uzbek}) to zero were made based solely on the spectral properties of the operators $\mathbf{\hat{B}}_{k}$ and on the initial environmental states $\boldsymbol{\hat{\rho}}^{E^{k},0}$ being supported on the \emph{Rajchman} subspace of $\mathbf{\hat{B}}_{k}$ respectively; this is the content of the final Theorem in \cite{AA}. Unfortunately, an analog to the latter can not be obtained within the parallel development for continuous variables presented in this work since in the bound provided for the diagonal terms is not independent of the minimization over all PVM. This being the case, it is not immediately clear what is the relationship between the spectral properties of the $\mathbf{\hat{B}}_{k}$ and the convergence of (\ref{eqn:uzbek}) to an SBS state (in the sense of Definition (\ref{eqn:turkcv}). Given that the decay of the off-diagonal terms (\ref{eqn:uzbek9}) as $t\rightarrow \infty$ can be cast in terms of the spectral properties of the dynamics generators corresponding to the environmental degrees of freedom that were traced out as seen at the beginning of this paragraph, it is tempting to hypothesize that the same conditions imposed on $\mathbf{\hat{B}}_{k}$ and the $\boldsymbol{\hat{\rho}}^{E^{k},0}$ in \cite{AA} leading to the convergence to SBS in discrete variables will be sufficient to guarantee that (\ref{eqn:uzbek}) decays to zero but proving this at the moment remains a challenge to the authors and therefore we leave it as a conjecture. What can be said conclusively about the diagonal terms is that they become zero if the respective supports of the density operators $\Lambda_{i,t}\Bigg(\bigotimes_{k=1}^{N_{E}}\boldsymbol{\hat{\rho}}^{E^{k},0}\Bigg)$ in Theorem \ref{eqn:themain2} are non-overlapping. Of course, this is an idealization and the best that one can hope for is that the overlap amongst the latter density operators (in the fidelity sense) becomes arbitrarily small as $t$ becomes arbitrarily large. However, our attempts to construct a time and initial data dependent PVM explicitly as was done in \cite{AA} for a generic setting have not been met with success; hence, the PVM used to analyze the asymptotic behavior of Theorems \ref{eqn:themain} and \ref{eqn:themain2} have to be constructed on a case-by-case basis.
\end{section}

\newpage
\appendix 
\addcontentsline{toc}{section}{Appendices}
\section{Quantum State Discrimination}
\label{app:QSD}
In this Appendix, we introduce the Quantum State Discrimination optimization problem (QSD)  \cite{hellstrom} \cite{Montanaro} \cite{qiu} \cite{bae} \cite{barnett}.  Let $\mathscr{H}$ be an arbitrary Hilbert space and let $\mathcal{S}(\mathcal{\mathscr{H}}\big)$ be the space of density operators acting in $\mathscr{H}$. Given a mixture of density operators, 
\begin{equation}
\label{eqn:countablemixture}
\boldsymbol{\hat{\rho}} = \sum_{i=1}^{N}p_{i} \boldsymbol{\hat{\rho}}_{i} 
\end{equation}
where $\sum_{i = 1}^{N}p_{i} = 1$, the theory of QSD aims to find a POVM $\{\mathbf{\hat{M}}_{l}^{\dagger}\mathbf{\hat{M}}_{l}\}_{l=1}^{K}\subset \mathcal{B}(\mathscr{H})$ ( $K\geq N$) which resolves the identity operator of $\mathcal{B}(\mathscr{H})$, and minimizes the object below which we will be referring to as a \emph{probability error}. 
\begin{equation}
\label{eqn:minerror2}
p_{E}\big\{\{p_{i},\boldsymbol{\hat{\rho}}_{i}\}_{i=1}^{N}, \{\mathbf{\hat{M}}_{l}\big\}_{l=1}^{K} \big\}:=1-\sum_{i=1}^{N}p_{i}Tr\big\{\mathbf{\hat{M}}_{i}\boldsymbol{\hat{\rho}}_{i}\mathbf{\hat{M}}^{\dagger}_{i} \big\} 
\end{equation}
\section{Multiple Partial Traces for Von Neumann Type Interaction}
\label{thelemmainthispape}
\begin{definition2}[Multiple Partial Traces]
\label{eqn:partialtrace}
\begin{equation}
Tr_{E^{N_{E}+1},E^{N_{E}+2},...,E^{N}}\big\{\boldsymbol{\hat{\rho}}_{t}\big\} =\mathscr{U}_{N_{E},t}\bigg( \mathscr{E}_{t}\big(\boldsymbol{\hat{\rho}}_{s}\big)\otimes\bigotimes_{k=1}^{N_{E}}\boldsymbol{\hat{\rho}}^{E^{k}_{0}}\bigg).
\end{equation}
where 
\begin{equation}
 \mathscr{U}_{n,t}\big( \mathbf{\hat{A}}\big) : = e^{-it \mathbf{\hat{X}}\otimes  \hat{\mathbf{S}}_{n}}\big( \mathbf{\hat{A}}\big)e^{it \mathbf{\hat{X}}\otimes\hat{\mathbf{S}}_{n}} 
\end{equation}
\begin{equation}
\hat{\mathbf{S}}_{n} := \sum_{k=1}^{n}g_{k}\mathbf{\hat{B}}_{k} 
\end{equation}
and
\begin{equation}
 \mathscr{E}_{t}^{M_{E}}\{ \boldsymbol{\hat{\sigma}}\} : =\int\int \langle x|\boldsymbol{\hat{\sigma}}|y\rangle\Gamma_{M_{E}}(t,x,y)|x\rangle\langle y| dxdy.
\end{equation}
Here
\begin{equation}
\label{eqn:scrmyonew}
\Gamma_{M_{E}}(t,x,y):=  \prod_{k=N_{E}+1}^{N}Tr_{k}\bigg\{e^{-itxg_{k}\mathbf{\hat{B}}_{k}} \boldsymbol{\hat{\rho}}^{E^{k}_{0}}e^{ityg_{k}\mathbf{\hat{B}}_{k}}\bigg\} 
\end{equation}
$M_{E}= N-N_{E}$, the number of traces being taken in equation (\ref{eqn:scrmyonew}). 
\end{definition2}
\begin{proof}
\begin{equation}
\label{eqn:tracemultiplenew}
Tr_{E_{N_{E}+1},E_{N_{E}+2},...,E_{N}}\big\{\boldsymbol{\hat{\rho}}_{t}\big\} =
\end{equation}
\begin{equation}
Tr_{E_{N_{E}+1},E_{N_{E}+2},...,E_{N}}\bigg\{e^{-it\mathbf{\hat{X}}\otimes\sum_{k=1}^{N}g_{k}\mathbf{\hat{B}}_{k}}\Bigg(\boldsymbol{\hat{\rho}}_{S_{0}}\otimes \bigotimes_{k=1}^{N}\boldsymbol{\hat{\rho}}^{E^{k}_{0}}\Bigg)e^{it\mathbf{\hat{X}}\otimes\sum_{k=1}^{N}g_{k}\mathbf{\hat{B}}_{k}}\bigg\}=
\end{equation}
\begin{equation}
\label{eqn:tokyonew}
\mathscr{U}_{N_{E},t}\Bigg(Tr_{E_{N_{E}+1},E_{N_{E}+2},...,E_{N}}\bigg\{e^{-it\mathbf{\hat{X}}\otimes\sum_{k=N_{E}+1}^{N}g_{k}\mathbf{\hat{B}}_{k}}\Bigg(\boldsymbol{\hat{\rho}}_{S_{0}}\otimes \bigotimes_{k=N_{E}+1}^{N}\boldsymbol{\hat{\rho}}^{E^{k}_{0}}\Bigg)e^{it\mathbf{\hat{X}}\otimes\sum_{k=N_{E}+1}^{N}g_{k}\mathbf{\hat{B}}_{k}}\bigg\}\bigotimes_{k=1}^{N_{E}}\boldsymbol{\hat{\rho}}^{E^{k}_{0}}\Bigg)
\end{equation}
Let us now use the generalized eigenvectors of $\mathbf{\hat{X}}$ in order to write $\boldsymbol{\hat{\rho}}_{S} = \int\int K_{S}(x,y)|x\rangle\langle y| dxdy$ where $K_{S}(x,y) = \langle x|\hat{\rho}_{S}|y\rangle$. We have
\begin{equation}
\label{eqn:okayamanew}
e^{-it\mathbf{\hat{X}}\otimes\sum_{k=N_{E}+1}^{N}g_{k}\mathbf{\hat{B}}_{k}}\Bigg(\boldsymbol{\hat{\rho}}_{S_{0}}\otimes \bigotimes_{k=N_{E}+1}^{N}\boldsymbol{\hat{\rho}}^{E^{k}_{0}}\Bigg)e^{it\mathbf{\hat{X}}\otimes\sum_{k=N_{E}+1}^{N}g_{k}\mathbf{\hat{B}}_{k}}=
\end{equation}
\begin{equation}
\int\int K_{S}(x,y)|x\rangle\langle y|\bigg(e^{-itx\sum_{k=N_{E}+1}^{N}g_{k}\mathbf{\hat{B}}_{k}} \Bigg(\bigotimes_{k=N_{E}+1}^{N}\boldsymbol{\hat{\rho}}^{E^{k}_{0}}\Bigg)e^{ity\sum_{k=N_{E}+1}^{N}g_{k}\mathbf{\hat{B}}_{k}}\bigg)dxdy=
\end{equation}
\begin{equation}
\int\int K_{S}(x,y)|x\rangle\langle y|\otimes\bigotimes_{k=N_{E}+1}^{N}e^{-itxg_{k}\mathbf{\hat{B}}_{k}}\boldsymbol{\hat{\rho}}^{E^{k}_{0}}e^{ityg_{k}\mathbf{\hat{B}}_{k}}dxdy.
\end{equation}
Furthermore 
\begin{equation}
Tr_{E_{N_{E}+1},E_{N_{E}+2},...,E_{N}}\bigg\{\int\int K_{S}(x,y)|x\rangle\langle y|\otimes\bigotimes_{k=N_{E}+1}^{N}e^{-itxg_{k}\mathbf{\hat{B}}_{k}}\boldsymbol{\hat{\rho}}^{E^{k}_{0}}e^{ityg_{k}\mathbf{\hat{B}}_{k}}dxdy\bigg\} = 
\end{equation}
\begin{equation}
 \int\int K_{S}(x,y)|x\rangle\langle y|Tr_{E_{N_{E}+1},E_{N_{E}+2},...,E_{N}}\bigg\{\bigotimes_{k=N_{E}+1}^{N}e^{-itxg_{k}\mathbf{\hat{B}}_{k}}\boldsymbol{\hat{\rho}}^{E^{k}_{0}}e^{ityg_{k}\mathbf{\hat{B}}_{k}}\bigg\}dxdy =
\end{equation}
\begin{equation}
\label{eqn:osakanew}
\int\int K_{S}(x,y)\Gamma_{M_{E}}(t,x,y)|x\rangle\langle y|dxdy= \mathscr{E}_{t}^{M_{E}}\big(\boldsymbol{\hat{\rho}}_{S_{0}}\big)
\end{equation}
Finally, using (\ref{eqn:tokyonew}) and (\ref{eqn:osakanew}), we have
\begin{equation}
\label{eqn:kyotonew}
(\ref{eqn:tokyonew})= \mathscr{U}_{N_{E},t}\Bigg(\mathscr{E}_{t}^{M_{E}}\big(\boldsymbol{\hat{\rho}}_{S_{0}}\big)\otimes\bigotimes_{k=1}^{N_{E}}\boldsymbol{\hat{\rho}}^{E^{k}_{0}}\Bigg)
\end{equation}
\end{proof}

\end{document}